\documentclass[12pt,a4paper]{article}

\usepackage[T1]{fontenc}
\usepackage{setspace}
\usepackage{bm}
\usepackage{varwidth}
\usepackage{amsthm}
\usepackage{mathrsfs}  
\usepackage{microtype}
\usepackage[hmargin=2.5cm, vmargin=2.5cm]{geometry}
\usepackage[inline]{enumitem}
\setlist{nosep} 
\usepackage[titletoc,title]{appendix}
\usepackage{graphicx}
\usepackage{amssymb,amsmath}
\usepackage{dsfont}
\usepackage[official]{eurosym}
\usepackage{natbib}
\usepackage[
            colorlinks=false,hidelinks,
            citecolor=black ]{hyperref} 
\usepackage{csquotes}            
\usepackage{caption}
\usepackage{pdfrender,xcolor}   
\usepackage{adjustbox}
\usepackage[flushleft]{threeparttable}
\usepackage{siunitx}
\usepackage{subcaption}
\usepackage{booktabs}
\usepackage{bigints}
\usepackage{color}
\color{black}
\newlist{mylist}{enumerate*}{1}
\setlist[mylist]{label=(\roman*)}   
        
\usepackage[most]{tcolorbox}
\newtcolorbox{highlighted}{colback=yellow,coltext=red,breakable}

\numberwithin{equation}{section}
\newtheorem{example}{Example}[section]

\newtheorem{theorem}{Theorem}[section]
\newtheorem{lemma}[theorem]{Lemma}
\newtheorem{proposition}[theorem]{Proposition}

\newtheorem{remark}{Remark}[section]

\makeatletter
\newcommand*{\rom}[1]{\expandafter\@slowromancap\romannumeral #1@}

\newcommand{\cond}{\;\middle\vert\;} 

\newcommand{\EE}{\mathds{E}} 

\newcommand{\eF}{\mathcal{F}}

\newcommand{\eL}{\mathcal{L}}

\newcommand{\PP}{\mathds{P}}  
\newcommand{\RR}{\mathds{R}}  

\newcommand{\Var}{\mathds{V}\kern-2pt\text{ar}} 
\newcommand{\Cov}{\mathds{C}\kern-2pt\text{ov}} 

\usepackage{xpatch}
\xpatchcmd{\qed}{\hfill}{}{}{}

\let\oldsqrt\sqrt
\def\sqrt{\mathpalette\APsqrt}
\def\APsqrt#1#2{%
\setbox0=\hbox{$#1\oldsqrt{#2}$}\dimen0=\ht0%
\advance\dimen0-0.2\ht0%
\setbox2=\hbox{\vrule height\ht0 depth -\dimen0}%
{\box0\lower0.48pt\box2}} 
\usepackage[font=footnotesize]{caption}

\allowdisplaybreaks
\begin{document}
{\singlespacing
\title{ \textbf{{\Large Dual Formulation of the Optimal Consumption Problem with Multiplicative Habit Formation}}\footnote{We are very grateful to An Chen, Tahir Choulli, Michel Vellekoop, Bas Werker, Ralf Werner and to the conference and seminar participants at the Netspar Pension Day (2021), Tilburg University (2021), Ulm University (2021), the International Pension Workshop (2022), the Winter Seminar on Mathematical Finance (2022), the $11^{\mathrm{th}}$ World Congress of the Bachelier Finance Society (2022), the $25^{\mathrm{th}}$ International Congress on Insurance: Mathematics and Economics (2022), the $5^{\mathrm{th}}$ European Actuarial Journal Conference (2022), the $14^{\mathrm{th}}$ Actuarial and Financial Mathematics Conference (2023) and the Doctoral Seminar in Sankt Gallenkirch (2023), for their helpful comments and suggestions.}}
\author{
    Thijs Kamma\footnote{Corresponding author. Mailing address: thijs.kamma@azl.eu. Akerstraat 92, 6411 HD Heerlen, the Netherlands. Phone: +31 (0)  88116 - 2000.}\\
    {\small Expert Center \& Product Data}\\
    {\small AZL}\\
    {\small NETSPAR}
    \and
    Antoon Pelsser\\
    {\small Dept. of Quantitative Economics}\\
    {\small Maastricht University}\\
    {\small  NETSPAR }
    \and
}

\date{\today}
\maketitle
\vspace{-0.75cm}
\begin{abstract}
{ \color{black}This paper provides a dual formulation of the optimal consumption problem with internal multiplicative habit formation. In this problem, the agent derives utility from the ratio of consumption to the internal habit component. Due to this multiplicative specification of the habit model, the optimal consumption problem is not strictly concave and incorporates irremovable path-dependency. As a consequence, standard Lagrangian techniques fail to supply a candidate for the corresponding dual formulation. Using Fenchel's Duality Theorem, we manage to identify a candidate formulation and prove that it satisfies strong duality. On the basis of this strong duality result, we are able to derive duality relations that stipulate how the optimal primal controls depend on the optimal dual controls and vice versa. {Moreover, using the dual formulation, we develop an analytical evaluation mechanism to bound the accuracy of approximations to the optimal solutions.}}
\end{abstract}
\noindent
{\footnotesize {\bf Keywords:} Fenchel Duality, habit formation, life-cycle investment,  stochastic optimal control, utility maximisation}\\
{\footnotesize {\bf JEL Classification:} C61, D15, D53, D81, G11}
\maketitle
\setcounter{page}{0}
\thispagestyle{empty}
}
%
\newpage
\section{Introduction}
Habit formation describes the phenomenon of an individual growing accustomed to a certain standard of living. In a financial context, this standard of living is dependent on a person's past decisions with regard to saving and consumption. Consuming more or less than a person-specific living standard may impact the utility levels of an individual, cf.~\citet{kahneman1979prospect}. It is therefore plausible that habit formation affects the current consumption behaviour of a person. To model and analyse the impact of habit-forming tendencies on this behaviour, a wide variety of studies have investigated optimal consumption problems that incorporate a habit level, representing the agent's living standard. These studies can be distinguished into two categories: (i) those that focus on additive habits, and (ii) those that concentrate on multiplicative habits. 

We start by discussing the additive habits. In optimal consumption problems with additive habits, the utility-maximising individual draws utility from the \textit{difference} between consumption and a habit level. The literature on these habits is pioneered by \citet{constantinides1990habit} and has been studied by e.g.~\citet{detemple1991asset}, \citet{campbell1999force}, \citet{munk2008portfolio}, \citet{muraviev2011additive}, and \citet{yu2015utility}. Additive habit models typically employ arithmetic habit levels, which monotonically increase over time, cf.~\citet{detemple2003non}, \citet{bodie2004optimal}, and \citet{polkovnichenko2007life}. Furthermore, as most standard utility functions only admit strictly positive arguments, additive habit specifications force the agent to maintain consumption above the habit level. For this reason, the habit component is sometimes interpreted as a subsistence level, see e.g.~\citet{yogo2008asset}. This interpretation is sensible for exogenous habits. However, if we assume that habits are endogenous, the habit level depends on the individual's past decisions and becomes person-specific. Consequently, for endogenous habits, it is hard not to consider the habit component as a standard of living that increases over time. 

Although individuals have a natural incentive to maintain consumption at least above their living standard, it is clear that additive habit models are too restrictive to be realistic. We attribute this restrictiveness to two main reasons. First of all, in practice, adverse changes in the financial circumstances can urge people to scale down consumption below the level to which they have become accustomed. Second, because of the latter phenomenon, an individual's standard of living may decrease over the course of a lifetime. To arrive at a more realistic model setup that manages to deal with the preceding two situations, the following two modifications can be made. As for the possibility of a declining standard of living, one can employ a geometric specification of the habit level, cf.~\citet{kozicki2002dynamic}, \citet{corrado2011multiplicative}, and \citet{van2020consumptiona}. Unlike the arithmetic habit levels, this geometric specification relies on the logarithmic transformation of consumption, and can therefore decrease over time. As for the possibility of scaling down consumption below the habit level, one can make use of multiplcative habit models. 

We now continue with a discussion of the multiplicative habits. Optimal consumption problems with multiplicative habit formation assume that the utility-maximising individual derives utility from the \textit{ratio} of consumption to a habit level. The specification of these habits dates back to \citet{abel1990asset}, and has been economically advocated by \citet{carroll2000solving} and \citet{carroll2000saving}. Contrary to the additive case, consumption is in this multiplicative setup not constrained to achieve values above the habit level. Namely, since the ratio of consumption to the habit level is always strictly positive, it can be included as an argument in all standard utility functions. The multiplicative habit model consequently allows the agent to reduce consumption levels below the habit component. Furthermore, the multiplicative habit model endows the utility-maximising agent with a strong incentive to fix consumption near/above the habit level. This incentive is due to the fact that the utility function of the agent increases with the magnitude of the ratio. 

When the habit level is endogenously determined (internal), standard solution techniques generally fail to solve optimal consumption problems with multiplicative habit formation in closed-form. Because of its dependence on past consumption decisions, the habit component gives rise to path-dependency in the objective function. This path-dependency is irremovable and cannot be handled in an analytical manner.\footnote{This analytical intractability is unique to problems involving multiplicative habits. In case of additive habits, the path-dependency can be eliminated from the problem, cf.~\citet{schroder2002isomorphism}.} Due to the structure of multiplicative habits, the optimal consumption problem is not strictly concave. In general, non-concave optimisation problems are more difficult to solve than concave ones, see e.g.~\citet{chen2019constrained}. To be able to analyse the corresponding optimal solutions, the general approach is to fall back on (i) numerical routines, (ii) approximations or (iii) duality techniques. In a discrete-time setup, \citet{fuhrer2000habit} and \citet{gomes2003portfolio} employ numerical methods to analyse the internal multiplicative habit model. More recently, in a continuous-time setup, \citet{van2020consumptiona}, \citet{li2021robust} and \citet{wang2024investment} have made use of both approximation and numerical routines.\footnote{We exclusively mention studies that focus on the consumption problem with internal multiplicative habits. Problems involving external habit formation, see e.g.~\citet{carroll1997comparison}, \citet{chan2002catching} and \citet{gomez2009implications}, do not pose issues when it comes to deriving optimal (duality) results. Martingale duality techniques, developed in the seminal contributions by \citet{pliska1986stochastic}, \citet{karatzas1987optimal}, and \citet{CoxHuang:1989:Optimalconsumptionand,cox1991variational}, suffice to analytically solve these consumption problems.} Although these studies provide valuable insights into the (optimal) solutions, they ignore potential benefits and insights from duality approaches. In fact, to the best of our knowledge, a dual formulation for the multiplicative habit model is not known. 

In this paper, we provide a dual formulation of the optimal consumption problem with internal multiplicative habit formation. We derive this formulation in a continuous-time setup with a finite trading-horizon. To this end, we rely on a multi-dimensional market model characterised by a general strictly positive semi-martingale process. For the agent's preferences, we make use of a utility function that admits a broad class of preference qualifications. The habit level of this utility-maximising individual is assumed to live by a geometric form. The conventional Lagrangian methods for obtaining dual formulations, e.g.~those in \citet{KleinRoger:2007:Dualityinoptimal} and \citet{rogers2003duality,rogers2013optimal}, are unable to supply a dual for this multiplicative habit problem. Namely, due to the fact that the problem is non-concave and involves path-dependency, the ordinary Legendre transform fails to establish the necessary conjugacy properties. {Therefore, we resort to an application of Fenchel's Duality Theorem. While Fenchel duality has a well-established history within the realm of mathematical finance, cf.~\citet{cui2018shortfall} and \citet{biagini2011admissible, biagini2020convex}, its implementation in problems with regard to habit formation is rare. By means of Fenchel duality} and a change of variables we are able derive a dual formulation and prove that strong duality holds. On the grounds of this strong duality outcome, we are able to derive so-called duality relations. These relations outline identities that disclose how the optimal primal controls depend on the optimal dual controls and vice versa. As a result, these identities infer something about the analytical structure of the optimal controls. 

{We illustrate the financial relevance of our duality result via a dual-control method. Similar methods have been developed for constrained investment-consumption problems, cf.~\citet{ma2017dual,ma2020dual} and \citet{weiss2020numerical}. Our method is inspired by \citet{bick2013solving} and \citet{kamma2021near}, and relates to the inability to provide closed-form solutions to the optimal consumption problem. Closed-form expressions for the optimal control(s) in frameworks with multiplicative habit formation are still unavailable. To the best of our knowledge, \citet{van2020consumptiona} is the only study that addresses this issue. By means of a log-linearisation procedure, the authors derive an analytical expression for an approximation to optimal consumption. They use a numerical grid-search routine in order to measure this accuracy. When the dimension of the problem is large, the estimated precision may be subject to a high degree of uncertainty. Our dual-control method is completely analytical and consequently tackles the latter problem. This duality-based method can be used for general approximations and generates a hard analytical upper bound on their accuracy. Due to the closed-form availability of the aforementioned upper bound, the method costs little to no computational effort.} 

The remainder of the paper is organised as follows. Section \ref{sec2} introduces the model setup and the optimal consumption problem with multiplicative habit formation. Section \ref{sec3} analyses the differences between additive and multiplicative habits on the subject of duality. Section \ref{sec44} contains a heuristic derivation of a candidate for the dual formulation using a modified version of \citet{KleinRoger:2007:Dualityinoptimal}'s identification method. Section \ref{sec4} presents and proves our main result: the dual formulation. In this section, we also elaborate on the duality relations corresponding to this result. {Finally, section \ref{sec6} concludes with a financial application to illustrate the economic relevance of our duality result. This application concerns the design of an analytical dual-control method.}

\section{Model Setup}\label{sec2}
In this section, we introduce the model setup. First, we lay out the financial market model. Second, we define the agent's wealth process. Third, we specify the agent's habit level. Fourth, we outline the optimal consumption problem.

\subsection{Financial Market Model}
Our financial market model is $N$-dimensional, defined in continuous-time, and based on the economic environments provided in \citet{detemple2010dynamic}, \citet{van2020consumptiona} and \citet{kamma2021near}. We define $T>0$ as the finite trading or planning horizon, and $\left[0,T\right]$ as the corresponding trading interval. Moreover, we introduce the complete filtered probability space $(\Omega,\eF,\left\lbrace\eF_t\right\rbrace_{t\in\left[0,T\right]},\PP)$. The components of this space live by their typical definitions, and its randomness is generated by an $\RR^{N}$-valued standard Brownian motion, $\left\lbrace W_t\right\rbrace_{t\in\left[0,T\right]}$. As of now, all (in)equalities between random variables and stochastic processes are understood in a $\PP$-a.s. or a $\mathrm{d}t\otimes \PP$-a.e. sense.

The financial market, $\mathcal{M}$, contains a scalar-valued money market account and $N$ risky assets that are represented by $N$ semi-martingale processes. The money market account submits to the following ordinary differential equation (ODE):
\begin{equation}
    \frac{\mathrm{d}B_t}{B_t}=r_t\mathrm{d}t, \ B_0=1.
\end{equation}
Here, $r_t$ represents the $\RR$-valued instantaneous interest rate. We assume that $r_t$ is $\eF_t$-progressively measurable and fulfills $\int_0^T\left|r_t\right|\mathrm{d}t<\infty$.  The price processes for the $N$ risky assets (stocks) evolve according to the following stochastic differential equation (SDE) for all $i=1,\hdots,N$:
\begin{equation}
\frac{\mathrm{d}S_{i,t}}{S_{i,t}}=\mu_{i,t}\mathrm{d}t+\sigma_{i,t}^{\top}\mathrm{d}W_t, \ S_{i,0}=1,
\end{equation}
where $\mu_{i,t}$ denotes the $\RR$-valued instantaneous expected return on stock $i$ and $\sigma_{i,t}$ the $\RR^N$-valued vector containing the volatility processes for stock $i$, both of which are $\eF_t$-progressively measurable. We postulate that $\int_0^T\left\Vert\mu_t\right\Vert_{\RR^N}\mathrm{d}t<\infty$ and $\int_0^T\mathrm{Tr}\big({\sigma}_t{\sigma}_t^{\top}\big)\mathrm{d}t<\infty$, in which $\mu_t\in\RR^N$ has entries $\mu_{i,t}$, and $\sigma_t\in\RR^{N\times N}$ rows $\sigma_{i,t}$, $i=1,\hdots,N$. Observe here that $\left\Vert\cdot\right\Vert_{\RR^N}$ denotes the $N$-dimensional Euclidean norm and that $\mathrm{Tr}\left(\cdot\right)$ represents the trace operator. To ensure invertibility of $\sigma_t$, we assume that $\sigma_t$ is non-singular.

Due to the absence of trading restrictions, this financial market is complete, i.e.~all traded risks are hedgeable. Hence, by the fundamental theorem of asset pricing, as formulated by \citet{Delbaen:Schachermayer:94}, there must exist a unique equivalent martingale measure. Correspondingly, there must exist a unique state price density (SPD), $\left\lbrace M_t\right\rbrace_{t\in\left[0,T\right]}$. Define $\lambda_t:=\sigma_t^{-1}\left(\mu_t-r_t1_N\right)$ as the market price of risk, then $M_t$ reads:
\begin{equation}\label{eq:sec2.1.spd}
\frac{\mathrm{d}M_t}{M_t}=-r_t\mathrm{d}t-\lambda_t^{\top}\mathrm{d}W_t, \ M_0=1.
\end{equation}
Note that $\left\lbrace B_t\right\rbrace_{t\in\left[0,T\right]}$ is selected as the num\'{e}raire quantity. We assume that $\left\lbrace\lambda_t\right\rbrace_{t\in\left[0,T\right]}$ satisfies $\EE\big[\exp\big(\frac{1}{2}\int_0^T\big\Vert{\lambda}_s\big\Vert_{\RR^N}^2\mathrm{d}s\big)\big]<\infty$, cf.~\citet{ks:stc}. Moreover, we postulate that $\left\lbrace\lambda_t\right\rbrace_{t\in\left[0,T\right]}$ is such that $M_t$ and $\log M_t$ attain values in $L^2\left(\Omega\times\left[0,T\right]\right)$.\footnote{We define $L^p\left(\Omega\times\left[0,T\right];\RR^n\right)$ as the standard Lebesgue space of all $\eF_t$-progressively measurable functions, $f:\Omega\times\left[0,T\right]\rightarrow\RR^n$, satisfying $\big(\int_{\Omega\times\left[0,T\right]}\left\Vert f_t\right\Vert^p_{\RR^n}\PP\left(\mathrm{d}t\right)\big)^{{1}/{p}}=\big(\EE\big[\int_0^T\left\Vert f_t\right\Vert^p_{\RR^n}\mathrm{d}t\big]\big)^{{1}/{p}}<\infty$. If $n=1$, we drop the ``$\RR$''-notation from the definition of the $L^p$ space.} The latter assumption is necessary to assure well-posedness of the dual formulation. In order to evaluate financial instruments in a risk-neutral fashion, one can make use of $M_t$. For example, $M_tB_t$ and $M_tS_t$ are both $\PP$-martingales with respect to $\left\lbrace\eF_t\right\rbrace_{t\in\left[0,T\right]}$.

\subsection{Dynamic Wealth Process}
In this environment, the agent is free to continuously select an investment and a consumption strategy over $\left[0,T\right]$. Specifically, the agent's wealth process, $\left\lbrace X_t\right\rbrace_{t\in\left[0,T\right]}$, is affected by two endogenous terms: (i) a process for the proportion of wealth that is allocated to the stock, $\left\lbrace \pi_t\right\rbrace_{t\in\left[0,T\right]}$, and (ii) a consumption process, $\left\lbrace c_t\right\rbrace_{t\in\left[0,T\right]}$. We assume that both preceding endogenous processes are $\eF_t$-progressively measurable. Let us fix a deterministic initial endowment, $X_0\in\RR_+$. Then, the agent's wealth process is defined by: 
\begin{equation}\label{eq:sec2.1bc}
\mathrm{d}X_t=X_t\left[\left(r_t+\pi_t^{\top}\sigma_t\lambda_t\right)\mathrm{d}t+\pi_t^{\top}\sigma_t\mathrm{d}W_t\right]-c_t\mathrm{d}t,
\end{equation} 

Clearly, $\left\lbrace c_t\right\rbrace_{t\in\left[0,T\right]}$ is $\RR_+$-valued and $\left\lbrace \pi_t\right\rbrace_{t\in\left[0,T\right]}$ is $\RR^N$-valued. A trading-consumption pair, $\left\lbrace c_t,\pi_t\right\rbrace_{t\in\left[0,T\right]}$, is said to be \textit{admissible} if it satisfies the following set of conditions: $X_t\geq 0$, $\int_0^T\pi_t^{\top}{\sigma}_t{\sigma}_t^{\top}\pi_t\mathrm{d}t<\infty$, $ \int_0^T\left|\pi_t^{\top}\sigma_t\lambda_t+r_tX_t\right|\mathrm{d}t<\infty$, and  $\log c_t\in L^2\left(\Omega\times\left[0,T\right]\right)$. The set containing all admissible trading-consumption pairs is denoted by $\mathcal{A}_{X_0}$. Observe that the proportion of wealth that is allocated to the cash account can be recovered from $1-\pi_t^{\top}1_N$, where $1_N$ is an $\RR^N$-valued vector containing only $1$'s. This specific proportion only plays a role through $\left\lbrace\pi_t\right\rbrace_{t\in\left[0,T\right]}$, due to which it can be excluded from the representation for $\left\lbrace X_t\right\rbrace_{t\in\left[0,T\right]}$. See e.g.~\citet{cuoco1997optimal} for a situation in which this is not the case.

\subsection{Habit Level}
The economic environment $\mathcal{M}$ consists of a utility-maximising agent who is internally habit-forming. As a consequence, the individual is in possession of a habit level, $h_t$ at time $t\in\left[0,T\right]$. This habit level represents the level of consumption to which the agent has become accustomed. Naturally, $h_t$ depends on the agent's preferences and his/her corresponding past consumption behaviour. Due to this dependence on past consumption decisions, the habit level constitutes an endogenous (internal) component. If $h_t$ is exogenously determined ($\beta=0$ below), the agent is externally habit-forming. By analogy with \citet{van2020consumptiona} and references therein, we suppose that the logarithmic transformation of this habit level, $h_t$, is given by:
\begin{equation}\label{eq:sec2.2.hablev}
\mathrm{d}\log h_t=\left(\beta\log c_t -\alpha\log h_t\right)\mathrm{d}t, \ \log h_0=0.
\end{equation}

The parameter $\beta\in\RR_+$ expresses the relative importance of past consumption decisions in the specification of $\log h_t$. For large values of $\beta$,  more weight is attached to these past consumption choices. For small values of $\beta$, the converse is true. The parameter $\alpha\in\RR_+$ stands for the habit level's rate of depreciation. For small values of $\alpha$, the habit level depends on past consumption decisions over a large time-horizon. For large values of $\alpha$, the converse is true. We assume that $\alpha\geq\beta$ holds, for concavity purposes related to the optimal consumption problem. The limiting case $\alpha=\beta=0$ results in $h_t=1$ for all $t\in\left[0,T\right]$. Setting $\alpha=\beta=0$ consequently recovers a model without habit formation. 

We note that the solution to the ODE in \eqref{eq:sec2.2.hablev} reads for all $t\in\left[0,T\right]$ as:
\begin{equation}\label{eq:sec2.2.habsol}
\log h_t = \beta\int_0^te^{-\alpha\left(t-s\right)}\log c_s\mathrm{d}s.
\end{equation} 
Hence, the habit level lives by a geometric form. That is, $h_t=\exp\big\{\beta\int_0^te^{-\alpha\left(t-s\right)}\log c_s\mathrm{d}s\big\}$ holds for all $t\in\left[0,T\right]$. In contrast with arithmetic habits, cf.~\citet{constantinides1990habit} and \citet{van2020consumptionms}, this specification of $h_t$ is not strictly increasing in time. As the geometric form consequently allows for decreases in $h_t$ over $t\in\left[0,T\right]$, the interpretation of this habit component as a standard of living is more sensible. Ultimately, we observe that $\log h_t$ in \eqref{eq:sec2.2.habsol} can be represented as follows: $\log h_t=\alpha\int_0^te^{-\alpha\left(t-s\right)}\log c_s^{{\beta}/{\alpha}}\mathrm{d}s$, for all $t\in\left[0,T\right]$. This representation indicates that $ h_t$ can be interpreted as the geometric weighted moving average (GWMA) of transformed past consumption decisions, $\big\{ c_s^{\beta/\alpha}\big\}_{s\in\left[0,t\right]}$. Clearly, if $\alpha=\beta\neq 0$, $\frac{c_t}{h_t}$ becomes a dimensionless quantity, and $h_t$ reduces to the ordinary GWMA of (non-transformed) past consumption decisions, $\left\lbrace c_s\right\rbrace_{s\in\left[0,t\right]}$. 

\subsection{Optimal Consumption Problem}
The habit-forming agent in $\mathcal{M}$ is at $t=0$ in possession of a predetermined amount of cash, $X_0\in\RR_+$, and lives until $t=T$. Throughout the trading interval, $\left[0,T\right]$, this agent seeks to maximise expected lifetime utility from the ratio of consumption to the habit process by continuously selecting his/her consumption levels and corresponding portfolio weights. The habit-forming agent must determine these controls in agreement with the dynamic budget constraint in \eqref{eq:sec2.1bc}, such that the admissibility conditions are met. We assume that the preferences of the individual are characterised by the nonseparable von Neumann-Morgenstern index: $\EE\big[\int_0^TU\left(t,{c_t}/{h_t}\right)\mathrm{d}t\big]$, cf.~\citet{detemple1991asset}. Consistent with this description, the agent faces the following problem:
\begin{equation}\label{eq:sec2.2dynprob}
\begin{aligned}
\sup_{\left\lbrace c_t,\pi_t\right\rbrace_{t\in\left[0,T\right]}\in{\mathcal{A}}_{X_0}} \ &\EE\left[\int_0^TU\left(t,\frac{c_t}{h_t}\right)\mathrm{d}t\right]\\ \mathrm{s.t.} \ \ & \mathrm{d}X_t=X_t\left[\left(r_t+\pi_t^{\top}\sigma_t\lambda_t\right)\mathrm{d}t+\pi_t^{\top}\sigma_t\mathrm{d}W_t\right]-c_t\mathrm{d}t, \\ & \mathrm{d}\log h_t=\left(\beta\log c_t -\alpha\log h_t\right)\mathrm{d}t, \ \ h_0=1, \ X_0\in\RR_+. 
\end{aligned}
\end{equation}

In this problem, $U:\left[0,T\right]\times\RR_+\rightarrow\RR$ denotes the agent's utility function. Henceforth, we set $U'\left(t,x\right)=\frac{\partial}{\partial x}U\left(t,x\right)$ and $U''\left(t,x\right)=\frac{\partial^2}{\partial x\partial x}U\left(t,x\right)$, for all $t\in\left[0,T\right]$ and $x\in\RR_+$. In line with the von Neumann-Morgenstern paradigm, we postulate that $U'\left(t,x\right)>0$ and $U''\left(t,x\right)<0$ hold, for all $t\in\left[0,T\right]$ and $x\in\RR_+$. Additionally, for purposes related to concavity of the optimisation problem, we assume that $-x\frac{U''\left(t,x\right)}{U'\left(t,x\right)}>1$, for all $t\in\left[0,T\right]$ and $x\in\RR_+$. This assumption is slightly stronger than the asymptotic elasticity requirement introduced by \citet{kramkov1999asymptotic}. Given the latter assumptions, we are able to define the unique function $I:\left[0,T\right]\times\RR_+\rightarrow\RR_+$, such that $U'\left(t,I\left(t,x\right)\right)I\left(t,x\right)=x$ holds, for all $t\in\left[0,T\right]$ and $x\in\RR_+$. Last, to ensure that $\EE\big[\int_0^TU\left(t,{c_t}/{h_t}\right)\mathrm{d}t\big]<\infty$ holds for all $\log c_t\in L^2\left(\Omega\times\left[0,T\right]\right)$, we postulate that there exists a constant, $M\in\RR_+$, such that $\int_0^TU\left(t,M\right)\mathrm{d}t<\infty$ and $\int_0^TU'\left(t,M\right)^2\mathrm{d}t<\infty$ hold true. 

\section{Habit Formation and Duality}\label{sec3}
In this section, we analyse two optimal consumption problems involving different forms of habit formation. One problem coincides with the \textit{multiplicative} formulation in \eqref{eq:sec2.2dynprob}. The other problem concerns the \textit{additive} counterpart of this formulation. Apart from the specification of the habit component and the argument of the utility function, these problems are identical to each other. We primarily concentrate on the derivation of the duals corresponding to these distinct problems. Thereby, we aim to highlight the non-trivial nature of the derivation required to obtain the dual formulation for multiplicative habit models. Subsequently, we first introduce the optimal consumption problem with additive habit formation. Second, we examine the most conventional approach to deriving the dual for this additive setup. Third, we demonstrate that this procedure fails to supply a well-defined dual formulation for the multiplicative problem presented in \eqref{eq:sec2.2dynprob}. 

\subsection{Consumption Problem with Additive Habits}\label{sec3.1}
We observe that the optimal consumption problem in \eqref{eq:sec2.2dynprob} incorporates a multiplicative habit component. To arrive at a formulation that corresponds to an additive habit model, we are required to moderately re-specify this problem. For this purpose, we adjust (i) the second argument of the utility function, and (ii) the definition of the habit component. We comment on the reasons for these primary adjustments in the sequel. On a more secondary level, for technical purposes, we replace the $\log c_t\in L^2\left(\Omega\times\left[0,T\right]\right)$ assumption in $\mathcal{A}_{X_0}$ by the following postulate: $c_t\in L^2\left(\Omega\times\left[0,T\right]\right)$. To avoid confusion, we denote the ensuing admissibility set by $\widehat{\mathcal{A}}_{X_0}$. Apart from these modifications, the notation from section \ref{sec2} as well as all the corresponding definitions and assumptions remain unchanged. The utility-maximising individual then faces the following optimal control problem:
\begin{equation}\label{eq:sec3.1.dynprob.additive}
\begin{aligned}
\sup_{\left\lbrace c_t,\pi_t\right\rbrace_{t\in\left[0,T\right]}\in{\widehat{\mathcal{A}}_{X_0}}} \ &\EE\left[\int_0^TU\left(t,c_t-h_t\right)\mathrm{d}t\right]\\ \mathrm{s.t.} \ \ & \mathrm{d}X_t=X_t\left[\left(r_t+\pi_t^{\top}\sigma_t\lambda_t\right)\mathrm{d}t+\pi_t^{\top}\sigma_t\mathrm{d}W_t\right]-c_t\mathrm{d}t, \\ & \mathrm{d} h_t=\left(\beta c_t -\alpha h_t\right)\mathrm{d}t, \ \ h_0=1, \ X_0\in\RR_+. 
\end{aligned}
\end{equation}

We first address the re-specified expected utility criterion. Compared to the agent in \eqref{eq:sec2.2dynprob}, the individual in \eqref{eq:sec3.1.dynprob.additive} derives utility from $c_t-h_t$ instead of $\frac{c_t}{h_t}$. Clearly, this modification is necessary to arrive at an additive habit model. In consideration of the domain of $\left(t,x\right)\mapsto U\left(t,x\right)$, the following constraint must be enforced upon the agent's consumption behaviour: $c_t>h_t$, for all $t\in\left[0,T\right]$. Put differently, the utility-maximising agent in \eqref{eq:sec3.1.dynprob.additive} is obliged to consume more than the habit level at all times. In order to make economic sense of this constraint, the literature interprets $h_t$ as a subsistence level, cf.~\citet{detemple1991asset}. This interpretation is sensible for large agents, e.g.~nations or large-scaled populations. However, adapted to smaller agents, it is considerably more intuitive to interpret $h_t$ as an individual-specific standard of living. In that case, the $c_t>h_t$ constraint is too restrictive to be realistic. We refer to the introduction for more arguments addressing the economic limitations of additive habit formation. Observe that \eqref{eq:sec2.2dynprob} is able to relax the former constraint.

Now, let us consider the re-definition of the habit component. Instead of a geometric habit level, the problem in \eqref{eq:sec3.1.dynprob.additive} includes an arithmetic habit component: 
\begin{equation}\label{eq:sec3.1.hab.comp}
    h_t=e^{-\alpha t}h_0 + \beta\int_0^te^{-\alpha\left(t-s\right)}c_s\mathrm{d}s,
\end{equation}
for all $t\in\left[0,T\right]$ and some $h_0\in\RR_+$. This modified specification of $h_t$ in the additive model is necessary for one important technical reason. Due to the discrepancy between $c_t$ and $\log c_t$, a geometric specification of $h_t$ results in problems with respect to the concavity of \eqref{eq:sec3.1.dynprob.additive}. Unlike in \eqref{eq:sec2.2dynprob}, this non-concavity cannot be eliminated from the problem by a change of variables. As a consequence, it is not possible to find a corresponding dual formulation. By means of the arithmetic habit component in \eqref{eq:sec3.1.dynprob.additive}, the problem is completely concave and therefore amenable to duality applications. For the same reason, the multiplicative habit problem incorporates a geometric habit component rather than an arithmetic one. Fenchel's Duality Theorem sheds further light on the interplay between $c_t$ and $h_t$ concerning the possibility to identify a dual formulation. Concretely, as long as one is able to specify the value function as some transformation of a bounded linear map of $c_t$ or $\log c_t$, a dual problem can be found, see Remark \ref{rem1}. Ultimately, we stress that the formulation in \eqref{eq:sec3.1.dynprob.additive} is standard, cf.~\citet{munk2008portfolio}. On the basis of our reasoning above, we can conclude that \eqref{eq:sec3.1.dynprob.additive} identifies the closest additive-based equivalent of \eqref{eq:sec2.2dynprob}. 

\begin{remark}\label{rem1}
In Proposition \ref{fenchelduality} of section \ref{appendix:a}, we supply the statement corresponding to Fenchel Duality. Under some regularity conditions, the proposition demonstrates that the problems $p^{*}$ and $d^{*}$ satisfy strong duality. For the interplay between $c_t$ and $h_t$, we must focus on the argument of $f^{*}$ in problem $d^{*}$ of \eqref{eq:fenchel.eq1}. Note that $f^{*}$ can be regarded as the objective function of \eqref{eq:sec2.2dynprob} or \eqref{eq:sec3.1.dynprob.additive}, In this argument, $A^{*}$ represents a bounded linear operator and $y^{*}$ the primal control (in our case $c_t$ or $\log c_t$). Therefore, the generality of Fenchel's strong duality result solely applies to objective functions involving a bounded linear operator of the primal control(s). Adapted to our situation, for example~$A^{*}y^{*}=c_t+\beta \int_0^te^{-\alpha \left(t-s\right)} c_s\mathrm{d}s$ is covered by Fenchel, whereas $A^{*}y^{*}=c_t+\beta \int_0^te^{-\alpha \left(t-s\right)} \log c_s\mathrm{d}s$ is not. Observe that standard (Lagrangian) duality results can be subsumed under the umbrella of Fenchel Duality. This phenomenon corroborates the need to specify $h_t$ for problems \eqref{eq:sec2.2dynprob} and \eqref{eq:sec3.1.dynprob.additive} in different ways. Under alternative habit specifications, we would not be able to derive corresponding dual formulations. For more details on Fenchel Duality, cf.~section \ref{appendix:a}.
\end{remark}

\subsection{Derivation of the Additive Dual}\label{sec3.2}
We continue with a derivation of the dual for optimal consumption problems involving additive habit formation. This derivation is based on well-documented conjugacy results and can be applied to a broad class of optimal control problems. For studies that strongly rely on these conjugacy properties in similar setups, see e.g.~\citet{cvitanic1992convex}, \citet{yu2015utility}, and \citet{czichowsky2016duality}. In this regard, it is noteworthy that the dual value function renders by definition an upper bound on the primal value function. The inequality inherent in this upper bound is generally predicated on the aforementioned conjugacy features. Indeed, in most utility-maximisation setups, the convex conjugate of the preference function manages to supply an inequality leading to the dual. On the basis of this conjugate function, we are also able to identify and derive a strong duality result for the additive problem in \eqref{eq:sec3.1.dynprob.additive}. For this reason, we introduce the convex conjugate corresponding to the utility function, $\left(t,x\right)\mapsto U\left(t,x\right)$, for all $t\in\left[0,T\right]$ and $x\in\RR_+$:
\begin{equation}\label{eq:sec3.2.conv.conj}
    V\left(t,x\right)=\sup_{z\in\RR_+}\left\lbrace U\left(t,z\right)-xz\right\rbrace,
\end{equation}

As a result of the convex conjugate, we are able to derive the following inequality: $U\left(t,x\right)\leq V\left(t,x\right)+xz$, for all $x,z\in\RR_+$ and $t\in\left[0,T\right]$. This inequality enables us to spell out an upper bound on the objective function of \eqref{eq:sec3.1.dynprob.additive}:
\begin{equation}\label{eq:sec3.2.ineq}
    \EE\left[\int_0^TU\left(t,c_t-h_t\right)\mathrm{d}t\right]\leq\EE\left[\int_0^TV\left(t,Z_t\right)\mathrm{d}t\right]+\EE\left[\int_0^TZ_t\left(c_t-h_t\right)\mathrm{d}t\right],
\end{equation}
for some $\RR_+$-valued process $Z:=\left\lbrace Z_t\right\rbrace_{t\in\left[0,T\right]}$. Although the right-hand side (RHS) of \eqref{eq:sec3.2.ineq} furnishes an upper bound on the primal value function, it is not necessarily independent of $\left\lbrace c_t\right\rbrace_{t\in\left[0,T\right]}$. By definition, a dual problem does not involve any primal control variables. Therefore, in its current form, the upper bound in \eqref{eq:sec3.2.ineq} does not qualify as a valid dual objective. Nevertheless, it is possible to identify $Z$ in such a manner that the preceding upper bound outlines a proper dual value function. To this end, let us inspect the definition of $Z$ more closely. As long as $Z$ does not depend on any primal controls, the first term on the RHS of \eqref{eq:sec3.2.ineq} is likewise independent of primal controls. This argument, however, does not carry over to the second term on the RHS of \eqref{eq:sec3.2.ineq}. Clearly, $\EE\left[\int_0^TZ_t\left(c_t-h_t\right)\mathrm{d}t\right]$ cannot be reduced to a primal-independent expression for all processes $Z$ that are independent of $\left\lbrace c_t\right\rbrace_{t\in\left[0,T\right]}$ and/or $\left\lbrace\pi_t\right\rbrace_{t\in\left[0,T\right]}$. To formulate the smallest possible upper bound on this expectation that is independent of primal processes, we employ the following results.  

Suppose that $\left\lbrace Y_t\right\rbrace_{t\in\left[0,T\right]}$ evolves according to:
\begin{equation}
    \frac{\mathrm{d}Y_t}{Y_t}=\gamma_t\mathrm{d}t+\delta_t^{\top}\mathrm{d}W_t,
\end{equation}
for some $Y_0\in\RR_+$, an $\RR$-valued process $\left\lbrace\gamma_t\right\rbrace_{t\in\left[0,T\right]}$, and an $\RR^N$-valued process $\left\lbrace\delta_t\right\rbrace_{t\in\left[0,T\right]}$, both of which are $\eF_t$-progressively measurable. We assume that $\left\lbrace\gamma_t\right\rbrace_{t\in\left[0,T\right]}$ and $\left\lbrace\delta_t\right\rbrace_{t\in\left[0,T\right]}$ are such that $\left\lbrace Y_t\right\rbrace_{t\in\left[0,T\right]}$ defines a semi-martingale process. Then, we can derive that:
\begin{equation}\label{eq:sec3.2.def.Y}
\int_0^TY_tc_t=X_0Y_0+\int_0^T\left(r_t+\pi_t^{\top}\sigma_t\left[\lambda_t+\delta_t\right]+\gamma_t\right)X_tY_t\mathrm{d}t+\int_0^T\left(\pi_t^{\top}\sigma_t+\delta_t\right)X_tY_t\mathrm{d}W_t.
\end{equation}
Here, we make use of the fact that $X_T=0$ should hold in problem \eqref{eq:sec3.1.dynprob.additive}. If $X_T>0$, more expected utility can be derived by consuming $c_t=X_t$ at $t=T$. Therefore, $X_T=0$ is a condition hidden under the surface of the optimisation problem in \eqref{eq:sec3.1.dynprob.additive}. Let us return to \eqref{eq:sec3.2.def.Y}, and note that $Y_0$, $\left\lbrace\gamma_t\right\rbrace_{t\in\left[0,T\right]}$ and $\left\lbrace\delta_t\right\rbrace_{t\in\left[0,T\right]}$ are to be determined. We aim to derive the smallest possible primal-independent upper bound on $\EE\left[\int_0^TY_tc_t\mathrm{d}t\right]$ by means of an appropriate specification of $\left\lbrace\gamma_t\right\rbrace_{t\in\left[0,T\right]}$ and $\left\lbrace\delta_t\right\rbrace_{t\in\left[0,T\right]}$. To this end, let us examine $\int_0^T\left(r_t+\pi_t^{\top}\sigma_t\left[\lambda_t+\delta_t\right]+\gamma_t\right)X_tY_t\mathrm{d}t$. Since $\left\lbrace \pi_t\right\rbrace_{t\in\left[0,T\right]}$ is not constrained, and $\left\lbrace r_t\right\rbrace_{t\in\left[0,T\right]}$ attains values in $\RR$, the smallest upper bound on the expectation of this integral is trivial for $\gamma_t\neq -r_t$ and $\delta_t\neq -\lambda_t$. Hence, in view of the goal of primal-independence, $\gamma_t=-r_t$ and $\delta_t=-\lambda_t$ must hold here. Note that $Y_t=Y_0M_t$, for $\gamma_t=-r_t$ and $\delta_t=-\lambda_t$, and all $t\in\left[0,T\right]$. Due to the latter specification of $\gamma_t$ and $\delta_t$, the drift term in \eqref{eq:sec3.2.def.Y} disappears and $\int_0^TY_tc_t=X_0Y_0+\int_0^T\left(\pi_t^{\top}\sigma_t+\delta_t^{\top}\right)X_tY_t\mathrm{d}W_t$ holds. That is, given these definitions, $\big\{\int_0^tY_sc_s\mathrm{d}s\big\}_{t\in\left[0,T\right]}$ qualifies a local $\PP$-martingale process. As this process is strictly positive, by Fatou's Lemma, we know that $\big\{\int_0^tY_sc_s\mathrm{d}s\big\}_{t\in\left[0,T\right]}$ outlines a supermartingale. Consequently, the smallest possible non-trivial upper bound on $\EE\left[\int_0^TY_tc_t\mathrm{d}t\right]$ is derived by fixing $Y_t=Y_0M_t$, for all $t\in\left[0,T\right]$:
\begin{equation}\label{eq:sec3.2.ineq2}
    \EE\left[\int_0^TY_tc_t\mathrm{d}t\right]=Y_0\EE\left[\int_0^TM_tc_t\mathrm{d}t\right]\leq X_0Y_0, \ \forall \ Y_0\in\RR_+.
\end{equation}

With this result at hand, let us consider $\EE\left[\int_0^TZ_t\left(c_t-h_t\right)\mathrm{d}t\right]$. Using a change in the order of integration, we are able to rewrite this term as follows:
\begin{equation}\label{eq:sec3.2.rewr.bc}
    \begin{aligned}
    \EE\left[\int_0^TZ_t\left(c_t-h_t\right)\mathrm{d}t\right]&=-h_0\EE\left[\int_0^TZ_te^{-\alpha t}\mathrm{d}t\right]\\&+\EE\left[\int_0^Tc_tZ_t\left\lbrace 1-\beta\EE\left[\int_t^T\frac{Z_s}{Z_t}e^{-\alpha\left(s-t\right)}\mathrm{d}s\cond\eF_t\right]\right\rbrace\mathrm{d}t\right].
    \end{aligned}
\end{equation}
We can focus on the second term on the RHS of \eqref{eq:sec3.2.rewr.bc}, as the first one is independent of primal controls for any primal-independent $Z$. Recall that primal-independence of $Z$ is necessary in order to be able to derive a well-defined problem objective in \eqref{eq:sec3.2.ineq}. In conformity with our analysis around the inequality in \eqref{eq:sec3.2.ineq2}, we know that the smallest non-trivial upper bound on the the former term can be derived by setting:
\begin{equation}\label{eq:sec3.2.volt.eq}
    Z_t-\beta\EE\left[\int_t^T{Z_s}e^{-\alpha\left(s-t\right)}\mathrm{d}s\cond\eF_t\right]=\eta M_t,
\end{equation}
for some $\eta\in\RR_+$ and all $t\in\left[0,T\right]$. The identity in \eqref{eq:sec3.2.volt.eq} outlines a Volterra equation for $Z$. The solution for $Z$ is accordingly for all $t\in\left[0,T\right]$ equal to:
\begin{equation}\label{eq:sec3.2.volt.eq.sol}
    Z_t=\eta \widehat{M}_t=\eta M_t\left\lbrace1+\beta\EE\left[\int_t^Te^{-\left(\alpha-\beta\right)\left[s-t\right]}\frac{M_s}{M_t}\mathrm{d}s\cond\eF_t\right]\right\rbrace.
\end{equation}

Using that these ingredients, from the inequality in \eqref{eq:sec3.2.ineq}, we are able to derive the following smallest non-trivial upper bound on $\EE\big[\int_0^TU\left(t,c_t-h_t\right)\mathrm{d}t\big]$:
\begin{equation}\label{eq:sec3.2.dual.ineq3}
\begin{aligned}
    \EE\left[\int_0^TU\left(t,c_t-h_t\right)\mathrm{d}t\right]&\leq \EE\left[\int_0^TV\left(t,\eta\widehat{M}_t\right)\mathrm{d}t\right]+\eta\EE\left[\int_0^T\widehat{M}_t\left(c_t-h_t\right)\mathrm{d}t\right]\\&\leq \EE\left[\int_0^TV\left(t,\eta\widehat{M}_t\right)\mathrm{d}t\right]+\eta\left(X_0-h_0\EE\left[\int_0^Te^{-\alpha t}\widehat{M}_t\mathrm{d}t\right]\right).
\end{aligned}
\end{equation}
Observe that $h_0\EE\big[\int_0^T\widehat{M}_te^{-\alpha t}\mathrm{d}t\big]=h_0\EE\big[\int_0^Te^{-\left(\alpha-\beta\right)t}M_t\mathrm{d}t\big]$ can be shown to hold true. The RHS of the second inequality in \eqref{eq:sec3.2.dual.ineq3} is clearly independent of any primal control variable, for all $\eta\in\RR_+$. For this reason, this expression qualifies as a valid objective for the dual formulation corresponding to \eqref{eq:sec3.1.dynprob.additive}. This brings us to the final step of the identification/dierivation procedure. Given a \textit{fixed} $\eta\in\RR_+$, the upper bound presented in \eqref{eq:sec3.2.dual.ineq3} is clearly the smallest possible one amongst the non-trivial options. However, the magnitude of the distance between the upper bound and the primal value function varies with respect to $\eta\in\RR_+$. Therefore, to find the smallest possible upper bound on $\EE\big[\int_0^TU\left(t,c_t-h_t\right)\mathrm{d}t\big]$ based on the RHS of \eqref{eq:sec3.2.dual.ineq3}, we must minimise $\EE\big[\int_0^TV\big(t,\eta\widehat{M}_t\big)\mathrm{d}t\big]+\eta\big(X_0-h_0\EE\big[\int_0^Te^{-\alpha t}\widehat{M}_t\mathrm{d}t\big]\big)$ over $\eta\in\RR_+$. Minimising this expression is precisely what the dual does. Hence, in accordance with the set of inequalities presented in \eqref{eq:sec3.2.dual.ineq3}, the dual formulation corresponding to \eqref{eq:sec3.1.dynprob.additive} ought to be as follows:
\begin{equation}\label{eq:sec3.2.dual.form}
    \inf_{\eta\in\RR_+}\left(\EE\left[\int_0^TV\left(t,\eta\widehat{M}_t\right)\mathrm{d}t\right]+\eta\left(X_0-h_0\EE\left[\int_0^Te^{-\alpha t}\widehat{M}_t\mathrm{d}t\right]\right)\right).
\end{equation}

We finalise this derivation by demonstrating that problems \eqref{eq:sec3.1.dynprob.additive} and \eqref{eq:sec3.2.dual.form} satisfy strong duality. For this purpose, we introduce a new function $\widehat{I}:\left[0,T\right]\times\RR_+\rightarrow\RR_+$, characterising the inverse of marginal utility: $U'\big(t,\widehat{I}\left(t,x\right)\big)=x$, for all $t\in\left[0,T\right]$ and $x\in\RR_+$. On the grounds of this function, it is easy to show that the following holds: $V\left(t,x\right)=U\big(t,\widehat{I}\left(t,x\right)\big)-x\widehat{I}\left(t,x\right)$, for all $t\in\left[0,T\right]$ and $x\in\RR_+$. Now, we note that minimisation of the objective in \eqref{eq:sec3.2.dual.form} results in the following first-order condition:
\begin{equation}\label{eq:sec3.2.foc}
\EE\left[\int_0^T\widehat{I}\left(t,\eta \widehat{M}_t\right)\widehat{M}_t\mathrm{d}t\right]=X_0-h_0\EE\left[\int_0^Te^{-\left(\alpha-\beta\right)t}M_t\mathrm{d}t\right].
\end{equation}
Due to monotonicity of $\left(t,x\right)\mapsto \widehat{I}\left(t,x\right)$ in its second argument, we know that there exists a unique $\eta\in\RR_+$ that satisfies the identity in \eqref{eq:sec3.2.foc}. Under the first-order condition in \eqref{eq:sec3.2.foc}, the optimised dual objective reads: $\EE\big[\int_0^TU\big(t,\widehat{I}\big(t,\eta^{\mathrm{opt}} \widehat{M}_t\big)\big)\mathrm{d}t\big]$, where $\eta^{\mathrm{opt}}\in\RR_+$ is such that \eqref{eq:sec3.2.foc} holds. The latter implies that strong duality holds if the following solution for $\left\lbrace c_t\right\rbrace_{t\in\left[0,T\right]}$ is admissible, for all $t\in\left[0,T\right]$:
\begin{equation}\label{eq:sec3.2.cand.sol}
    c_t-h_t=\widehat{I}\left(t,\eta^{\mathrm{opt}}\widehat{M}_t\right).
\end{equation}
Henceforth, we assume that $M_t$ is such that $\widehat{I}\big(t,\eta^{\mathrm{opt}}\widehat{M}_t\big)\in L^2\left(\Omega\times\left[0,T\right]\right)$. To prove that $\left\lbrace c_t,\pi_t\right\rbrace_{t\in\left[0,T\right]}\in\widehat{\mathcal{A}}_{X_0}$, let us define $\widehat{c}_t:=c_t-h_t$, for all $t\in\left[0,T\right]$. Then, it can be shown that the following holds: $h_t=e^{\left(\beta-\alpha\right)t}h_0+\beta\int_0^te^{\left(\beta-\alpha\right)\left[t-s\right]}\widehat{c}_s\mathrm{d}s$, for all $t\in\left[0,T\right]$. Hence, $\left\lbrace c_t\right\rbrace_{t\in\left[0,T\right]}$ can be uniquely retrieved from \eqref{eq:sec3.2.cand.sol} according to: $c_t=\widehat{c}_t+e^{\left(\beta-\alpha\right)t}h_0+\beta\int_0^te^{\left(\beta-\alpha\right)\left[t-s\right]}\widehat{c}_s\mathrm{d}s$, where $\widehat{c}_t=\widehat{I}\big(t,\eta^{\mathrm{opt}}\widehat{M}_t\big)$, for all $t\in\left[0,T\right]$. The condition imposed upon $\widehat{c}_t$ in \eqref{eq:sec3.2.cand.sol} by $\eta^{\mathrm{opt}}$ is given in \eqref{eq:sec3.2.foc}. Combining \eqref{eq:sec3.2.foc} and \eqref{eq:sec3.2.cand.sol}, the following should hold: $\EE\big[\int_0^T\widehat{c}_t\widehat{M}_t\big]=X_0-h_0\EE\big[\int_0^Te^{-\left(\alpha-\beta\right)t}M_t\mathrm{d}t\big]$. The latter identity reduces to $\EE\big[\int_0^TM_tc_t\mathrm{d}t\big]=X_0$, using the definition of $h_t$ involving $\big\{\widehat{c}_s\big\}_{s\in\left[0,t\right]}$ and the fact that $\widehat{c}_t=c_t-h_t$. On the grounds of the preceding equation, we know that $\big\{\int_0^tc_sM_s\mathrm{d}s\big\}_{t\in\left[0,T\right]}$ outlines a $\PP$-martingale process with respect to $\left\lbrace\eF_t\right\rbrace_{t\in\left[0,T\right]}$. Consequently, resorting to hedging arguments, there exists a trading strategy $\left\lbrace\pi_t\right\rbrace_{t\in\left[0,T\right]}$ corresponding to $\left\lbrace c_t\right\rbrace_{t\in\left[0,T\right]}$ implied by \eqref{eq:sec3.2.cand.sol}. By Lemma 2.2 in \citet{CoxHuang:1989:Optimalconsumptionand} and Proposition 7.3 in \citet{cvitanic1992convex}, we therefore have that $\left\lbrace c_t,\pi_t\right\rbrace_{t\in\left[0,T\right]}\in\widehat{\mathcal{A}}_{X_0}$. Since the consumption process inherent in \eqref{eq:sec3.2.cand.sol} is admissible, the problems in \eqref{eq:sec3.1.dynprob.additive} and \eqref{eq:sec3.2.dual.form} satisfy strong duality. Note that the formulation in \eqref{eq:sec3.2.dual.form} coincides with the dual for the unconstrained case provided in \citet{yu2015utility}.\footnote{For ease of exposition, we have restricted ourselves to the most standard specification of the optimal consumption problem with additive habit formation. The relative simplicity of this formulation manages to highlight the crucial steps associated with typical duality derivations. While the inclusion of endowment streams, convex trading constraints, and/or terminal wealth criteria may enlarge the exposition's generality, the core steps remain more or less the same. For a comprehensive duality-oriented theoretical treatment of a significantly more general problem, cf.~\citet{yu2015utility}. The author examines the same formulation in an incomplete market with a pre-defined (partially) non-traded labour income process. Note that the choice for the structure of problem \eqref{eq:sec3.1.dynprob.additive} is also motivated by the technical specification of \eqref{eq:sec2.2dynprob}.} 

\subsection{Issues with Multiplicative Habits}\label{sec3.3}
In this section, we analyse the issues that arise when attempting to derive the dual corresponding to \eqref{eq:sec2.2dynprob} by means of conventional machinery. For this purpose, we rely on the approach from the previous section. As this approach manages to generate duals for a great majority of convex optimisation problems, we are able to emphasise the non-trivial nature of our main result, i.e.~the dual formulation in Theorem \ref{thm4.1}. Due to the similarities between problems \eqref{eq:sec2.2dynprob} and \eqref{eq:sec3.1.dynprob.additive}, our review of the derivation technique in section \ref{sec3.2} proves useful. From the analyses around \eqref{eq:sec3.2.conv.conj}, we recall that the dual formulation is required to spawn an upper bound on the primal value function. In an attempt to derive this upper bound, we employ the inequality inherent in the convex conjugate function $\left(t,x\right)\mapsto V\left(t,x\right)$. This step of the conventional derivation corresponds to the one presented in \eqref{eq:sec3.2.ineq}. As before, we introduce an $\RR_+$-valued semi-martingale process $Z:=\left\lbrace Z_t\right\rbrace_{t\in\left[0,T\right]}$, to be determined. Note that this process does not coincide with $Z$ from section \ref{sec3.2}. With these ingredients at hand, we derive the following upper bound: 
\begin{equation}\label{eq:sec3.3.ineq1}
    \EE\left[\int_0^TU\left(t,\frac{c_t}{h_t}\right)\mathrm{d}t\right]\leq \EE\left[\int_0^TV\left(t,Z_t\right)\mathrm{d}t\right]+\EE\left[\int_0^TZ_t\frac{c_t}{h_t}\mathrm{d}t\right].
\end{equation}

In addition to supplying an upper bound, the dual objective should also be independent of all primal controls. Following the situation for additive habits, the first term on the RHS of \eqref{eq:sec3.3.ineq1} is primal-independent as long as $Z$ does not depend on $\left\lbrace c_t\right\rbrace_{t\in\left[0,T\right]}$ and $\left\lbrace \pi_t\right\rbrace_{t\in\left[0,T\right]}$. We consequently assume that $Z$ is primal-independent as well. As for the second term on the RHS of \eqref{eq:sec3.3.ineq1}, the latter assumption does not ensure that $\EE\big[\int_0^TZ_t\frac{c_t}{h_t}\mathrm{d}t\big]$ is independent of the primal controls. In fact, we are able to demonstrate that there does not exist a primal-independent process $Z$ such that $\EE\big[\int_0^TZ_t\frac{c_t}{h_t}\mathrm{d}t\big]$ or some non-trivial bound on this expression is likewise primal-independent. Clearly, the non-existence of such a process directly encumbers a derivation of the dual formulation corresponding to \eqref{eq:sec2.2dynprob} using the conventional approach. To show that there does not exist a such a process, let us inspect the integral in the second term on the RHS of \eqref{eq:sec3.3.ineq1} more closely. Suppose that $\frac{\mathrm{d}Z_t}{Z_t}=\kappa_t\mathrm{d}t+\xi_t^{\top}\mathrm{d}W_t$, for some $Y_0\in\RR_+$, an $\RR$-valued process $\left\lbrace\gamma_t\right\rbrace_{t\in\left[0,T\right]}$, and an $\RR^N$-valued process $\left\lbrace\delta_t\right\rbrace_{t\in\left[0,T\right]}$, both of which are $\eF_t$-progressively measurable. Given that $X_T=0$ is a latent constraint in \eqref{eq:sec2.2dynprob}, we are then able to derive that the following holds: 
\begin{equation}\label{eq:sec3.3.integr.expr}
    \begin{aligned}
        \int_0^TZ_t\frac{c_t}{h_t}\mathrm{d}t&=X_0Z_0+\int_0^T\frac{X_tZ_t}{h_t}\left(\pi_t^{\top}\sigma_t+\xi_t^{\top}\right)\mathrm{d}W_t\\ +&\int_0^T\frac{X_tZ_t}{h_t}\left(r_t+\pi_t^{\top}\sigma_t\left[\lambda_t+\xi_t\right]+\kappa_t-\beta\log c_t+\alpha\log h_t\right)\mathrm{d}t.
    \end{aligned}
\end{equation}
In the spirit of the inequality in \eqref{eq:sec3.3.ineq1}, the objective is now to derive a meaningful upper bound on the expectation of the preceding expression. For this reason, let us focus on the third term on the last line of \eqref{eq:sec3.3.integr.expr}. As in additive case, for general $\kappa_t$ and $\xi_t$, it is not possible to present a non-redundant upper bound on the expectation of \eqref{eq:sec3.3.integr.expr}. Due to the fact that $\left\lbrace r_t\right\rbrace_{t\in\left[0,T\right]}$, $\left\lbrace \log c_t\right\rbrace_{t\in\left[0,T\right]}$ and $\left\lbrace \log h_t\right\rbrace_{t\in\left[0,T\right]}$ achieve values in $\RR$, and that $\pi_t$ is unconstrained, the smallest upper bound is trivial. However, unlike in the derivation for additive habits, we are not able to define the processes $\left\lbrace\kappa_t\right\rbrace_{t\in\left[0,T\right]}$ and $\left\lbrace\xi_t\right\rbrace_{t\in\left[0,T\right]}$ such that this integral can be properly bounded. That is, there do not exist primal-independent processes $\left\lbrace\kappa_t\right\rbrace_{t\in\left[0,T\right]}$ and $\left\lbrace\xi_t\right\rbrace_{t\in\left[0,T\right]}$, which ensure that $r_t+\pi_t^{\top}\sigma_t\left[\lambda_t+\xi_t\right]+\kappa_t-\beta\log c_t+\alpha\log h_t=0$ holds, for all $t\in\left[0,T\right]$. Note that nullity of this integrand is the mere way to guarantee that the expectation of \eqref{eq:sec3.3.integr.expr} is bounded in a non-trivial sense. The inability to satisfy this nullity condition is entirely attributable to the inclusion of the $\log c_t$ and $\log h_t$ terms in the integrand expression. These terms follow as a result of the multiplicative habit structure. For the preceding reason(s), it is not possible to specify a well-defined dual formulation for \eqref{eq:sec2.2dynprob} by means of the conventional derivation procedure. 

\section{Heuristic Derivation}\label{sec44}
In this section, we present a heuristic derivation of the dual corresponding to \eqref{eq:sec2.2dynprob} using an adjusted version of \citet{KleinRoger:2007:Dualityinoptimal}'s identification method. We elaborate on this identification technique for two reasons. First, the standard variant of their identification method does not give rise to a well-posed dual value function. In line with our alternative derivation of the dual, we have to slightly modify their procedure. Second, the successful heuristic identification stresses that our unconventional duality result can be traced back to fairly conventional techniques. Subsequently, we start by showing that a standard application of the identification method furnishes an ill-posed dual problem. Thereafter, we demonstrate that a mildly modified equivalent of this method manages to provide a proper candidate for the dual formulation. For later reference, we note that the non-adjusted identification procedure distinguishes the following five steps:\\
  \begin{center}
    \begin{minipage}{0.85\textwidth}
      \begin{enumerate}[label=(\roman*)]
    \item Introduce Lagrangian (exponential) semi-martingale processes.
    \item Apply integration-by-parts to the Lagrangian processes and constraints.
    \item Construct the Lagrangian functional.
    \item Formulate complementary slackness (CS) conditions.
    \item Minimise the ensuing expression over the Lagrangian processes.\\
\end{enumerate}
    \end{minipage}
  \end{center}

\subsection{Standard Identification Method}\label{sec4.1}
We continue with an application of the standard identification method to the problem in \eqref{eq:sec2.2dynprob}. Note that this method can be considered as a Lagrangian framework adapted to optimal control problems in continuous-time. For this standard procedure, we consecutively follow the previous five steps. Therefore, in line with step (i), we start by introducing two semi-martingales, $Z^0_t$ and $Z^1_t$. These semi-martingales serve to enforce the dynamics upon the processes $X_t$ and $h_t$. As $X_t$ and $h_t$ both achieve values in $\RR_+$, we postulate an exponential form for these two processes. The SDE's of $Z^0_t$ and $Z^1_t$ accordingly read: 
\begin{equation}\label{eq:sec3.1.zz}
    \begin{aligned}
    \mathrm{d}Z^0_t&=Z^0_t\left[a_t\mathrm{d}t+b_t^{\top}\mathrm{d}W_t\right]\\
    \mathrm{d}Z^1_t&=Z^1_t\left[\phi_t\mathrm{d}t+\chi_t^{\top}\mathrm{d}W_t\right],
    \end{aligned}
\end{equation}
for some $Z^0_0,Z^1_0\in\RR_+$. The drift and diffusion terms of these two processes, i.e.~$a_t$, $b_t$, $\phi_t$ and $\chi_t$, as well as their starting values, $Z^0_0$ and $Z^1_0$, are (partially) determined as part of the identification process. As in \citet{KleinRoger:2007:Dualityinoptimal}, we note that $Z^0_t$ and $Z^1_t$ can be regarded as ordinary Lagrange multipliers for the constraints in \eqref{eq:sec2.2dynprob}. This is intuitive, if we view the dynamics as constraints that should be satisfied by $X_t,h_t,c_t$ and $\pi_t$.

In step (ii), we compute the dynamics of $X_tZ^0_t$ and $h_tZ^1_t$, using It\^{o}'s Lemma or integration-by-parts. The SDE of the product process $X_tZ^0_t$ is given by:
\begin{equation}\label{eq:sec3.1.xzdyn}
    \begin{aligned}
    \mathrm{d}X_tZ^0_t&=X_t\mathrm{d}Z^0_t+Z^0_t\mathrm{d}X_t+\mathrm{d}\left\langle X,Z^0\right\rangle_t\\&=X_tZ^0_t\left[\left(r_t+\pi_t^{\top}\sigma_t\lambda_t\right)\mathrm{d}t+\pi_t^{\top}\sigma_t\mathrm{d}W_t\right]-c_tZ^0_t\mathrm{d}t\\&+X_tZ^0_t\left[a_t\mathrm{d}t+b_t^{\top}\mathrm{d}W_t\right]+\pi_t^{\top}\sigma_t b_tX_tZ^0_t\mathrm{d}t.
    \end{aligned}
\end{equation}
Likewise, the SDE of the product process $h_t Z^1_t$ follows:
\begin{equation}\label{eq:sec3.1.xzdyn.h}
    \begin{aligned}
    \mathrm{d}Z^1_t h_t&=h_t\mathrm{d}Z^1_t+Z^1_t\mathrm{d} h_t+\mathrm{d}\left\langle h,Z^1\right\rangle_t\\&=h_tZ^1_t\left[\beta\log c_t -\alpha\log h_t+\phi_t \right]\mathrm{d}t+\chi_t^{\top}h_tZ^1_t\mathrm{d}W_t.
    \end{aligned}
\end{equation}
We assume that the diffusion terms in the preceding two processes are regular-enough, such that their stochastic integrals vanish in expectation. As pointed out by \citet{KleinRoger:2007:Dualityinoptimal}, this assumption has to be justified in the proof of the dual. In order to make the dynamics of the processes in \eqref{eq:sec3.1.xzdyn} and \eqref{eq:sec3.1.xzdyn.h} amenable to Lagrangian machinery, we (i) solve the SDE's at $t=T$, and (ii) take expectations on both sides of the resulting equations. By assumption, integration with respect to $\PP$ eliminates the two stochastic integrals. After this, we rewrite the identities such that the expectations equate to zero. In this way, the resulting expressions can be incorporated into the primal objective function in \eqref{eq:sec2.2dynprob}. For the product process $X_TZ^0_T$, these operations result in the following identity: 
\begin{equation}\label{eq:sec3.1.exp1}
    \begin{aligned}
    \EE\left[-X_TZ^0_T+\int_0^T\left(r_t+a_t+\pi_t^{\top}\sigma_t\left[\lambda_t+b_t\right]\right)X_tZ^0_t\mathrm{d}t-\int_0^Tc_tZ^0_t\mathrm{d}t\right]+Z^0_0X_0=0.
    \end{aligned}
\end{equation}
For $h_TZ^1_T$, we arrive at the following equation:
\begin{equation}\label{eq:sec3.1.exp2}
    \begin{aligned}
    \EE\left[h_TZ^1_T-\int_0^T\left[\beta\log c_t-\alpha\log h_t+\phi_t\right] Z^1_th_t\mathrm{d}t\right]-Z^1_0h_0=0.
    \end{aligned}
\end{equation}
For the identification of the dual, we remark that it makes no difference whether we multiply these equations by $-1$ or not, as this will affect both $Z^0_t$ and $Z^1_t$ at step (iv). 

In step (iii), we compute the Lagrangian functional corresponding to the optimal consumption problem in \eqref{eq:sec2.2dynprob}. In specifying this functional, we adhere to the ordinary Lagrangian principles. With the two identities in \eqref{eq:sec3.1.exp1} and \eqref{eq:sec3.1.exp2} at hand, it is fairly easy to construct the aforementioned functional. In fact, the Lagrangian develops from inserting the aforementioned identities into the primal objective, and subsequently optimising the ensuing expression over all endogenous processes $\left(c_t,h_t,X_t,\pi_t\right)\in\RR^3_+\times\RR^N$. These operations yield us the following expression for the Lagrangian, $\eL$:
\begin{equation}\label{eq:sec3.1.lagr}
\begin{aligned}
    \eL=\sup_{\left\lbrace {c}_t\right\rbrace,\left\lbrace {h}_t\right\rbrace,\left\lbrace X_t\right\rbrace,\left\lbrace \pi_t\right\rbrace}&\EE\Big[-Z^0_TX_T+Z^1_Th_T+\int_0^T\bigg\{U\left(t,\frac{{c}_t}{{h}_t}\right)\\&+\left(r_t+a_t+\pi_t^{\top}\sigma_t\left[\lambda_t+b_t\right]\right)X_tZ^0_t-c_tZ^0_t\\&-\left\lbrace\beta\log c_t-\alpha\log h_t+\phi_t\right\rbrace Z^1_th_t\bigg\}\mathrm{d}t\Big]+Z^0_0X_0-Z^1_0h_0.
\end{aligned}
\end{equation}

In step (iv), we determine the CS conditions corresponding to the optimisation problem outlined by $\eL$. Clearly, the characterisation of the CS conditions involves the specification of (local) maximisers of $\eL$. As for the CS conditions, we observe that maximisation over $X_T,h_T\geq 0$ results in a finite-valued Lagrangian functional, $\eL$, if $X_T=0$ and $Z^1_T=0$ hold true. Analogously, maximising $\eL$ over $X_t>0$ and $\pi_t\in\RR^N$, results in $\eL<\infty$, if the following conditions hold: $a_t\leq -r_t$ and $b_t=-\lambda_t$. To finalise this step, we have to compute $c_t$ and $h_t$ that optimise the objective function of $\eL$. Using elementary differentiation techniques, we find the following two first-order conditions (FOC's): 
\begin{equation}\label{eq:sec3.1.sys}
    \begin{aligned}
    &U'\left(t,\frac{{c}_t}{{h}_t}\right)\frac{1}{{h}_t}-Z_t^0-\beta\frac{h_t}{{c}_t}Z_t^1=0,\\ &U'\left(t,\frac{{c}_t}{{h}_t}\right)\frac{1}{{h}_t}+\left[\beta\log c_t-\alpha\log h_t+\phi_t\right] \frac{h_t}{c_t}Z^1_t -\alpha\frac{h_t}{c_t} Z_t^1=0.
    \end{aligned}
\end{equation}

Ordinarily, the identification procedure would continue by inserting the solutions to \eqref{eq:sec3.1.sys} into $\eL$. The resulting expression for $\eL$ is called the ``reduced Lagrangian''. In step (v), the candidate dual formulation is then identified as the problem in which the reduced Lagrangian is minimised over all $\left(Z^0_t,Z^1_t\right)$, such that the CS conditions are met. There are, however, two issues involved with these steps. The first issue is the impossibility to solve \eqref{eq:sec3.1.sys} in closed-form.\footnote{From \eqref{eq:sec3.1.sys}, for $\alpha=\beta$, we can derive that $ \frac{c_t}{h_t}=\beta \big({Z^1_t}/{Z^0_t}\big)W\big({e^{-{\phi_t}/{\beta}}Z^0_t}/\big({\beta Z^1_t}\big)\big)$ holds for all $t\in\left[0,T\right]$, where $W\left(\cdot\right)$ represents the Lambert function. By virtue of the intractability of the Lambert function, it is clear that neither $c_t$ nor $h_t$ can be obtained from the system of equations in closed-form. For the $\alpha\neq\beta$ case, it is not possible to derive such a semi-analytical expression for $\frac{c_t}{h_t}$.} Consequently, we cannot arrive at an analytically defined dual candidate. The second, more fundamental issue is that this dual candidate does not engender an upper bound on the primal value function. 

To be able to see this, let us introduce the function $L:\RR^2_+\rightarrow\RR$, given by: 
\begin{equation}\label{eq:sec3.1.funcl}
    L\left(x,y\right)=U\left(t,\frac{x}{y}\right)-xz_0-\left[\beta\left(\log x - \log y\right)+\phi\right]z_1 y,
\end{equation}
where we hold the constants $z_0,z_1,t\in\RR_+$ and $\phi\in\RR$ fixed. It is not difficult to show that the determinant of the Hessian of $L$, say $\mathrm{det}\left(H_L\right)$, is given by: $\mathrm{det}\left(H_L\right)=-U'\big(t,\frac{x}{y}\big)^2\frac{1}{y^4}+\left[\alpha-\beta\right]U''\left(t,\frac{x}{y}\right)\frac{z_1}{y_3}+\beta\left[\alpha-\beta\right]z_1^2\frac{1}{x^2}$. Now, we observe that:
\begin{equation}
\begin{aligned}
    \frac{\partial L}{\partial x}&=U'\left(t,\frac{x}{y}\right)\frac{1}{y}-z_0-\beta z_1\frac{y}{x},\\ \frac{\partial L}{\partial y}&=-U'\left(t,\frac{x}{y}\right)\frac{x}{y^2}-\left[\beta\left(\log x - \log y\right)+\phi\right]+\beta z_1.
\end{aligned}
\end{equation}
characterise the partial derivatives of $L$ in the $x$- and $y$-directions, respectively. The function $L$ attains a stationary point at $\left(x^{*},y^{*}\right)$, which ensues from solving $\frac{\partial L}{\partial x}=\frac{\partial L}{\partial y}=0$ for $\left(x,y\right)$. At $\left(x,y\right)=\left(x^{*},y^{*}\right)$, $\frac{\partial L}{\partial x\partial x}<0$ and $\mathrm{det}\left(H_L\right)<0$ both hold true. Therefore, we know that $\left(x^{*},y^{*}\right)$ describes a saddle point of $L$.  Observe that $\mathrm{det}\left(H_L\right)<0$ for \textit{all} $x,y\in\RR_+$ if $\alpha=\beta$. Noting that the FOC's in \eqref{eq:sec3.1.sys} correspond to the state-wise stationary point of $L\left(c_t,h_t\right)$ for $z_0=Z^0_t$, $z_1=Z^1_t$ and $\phi=\phi_t$, at all $t\in\left[0,T\right]$, we can conclude that the solution to the identities in \eqref{eq:sec3.1.sys} neither locally nor globally maximise $\eL$'s objective. We emphasise that fixing $z_0=Z^0_t$, $z_1=Z^1_t$ and $\phi=\phi_t$ is valid, as these processes are unaffected by $\left(c_t,h_t\right)$. As a result, we know that the non-analytical expression for the reduced Lagrangian, $\eL$, does not render an upper bound on the primal value function. Therefore, by means of the standard variant of \citet{KleinRoger:2007:Dualityinoptimal}'s identification method, we cannot derive a candidate for the dual corresponding to \eqref{eq:sec2.2dynprob}.

\begin{remark}\label{rem2}
Concerning the preceding identification, we wish to make three remarks. First, suppose that $Z^1_t$ is characterised by a general semi-martingale process:
\begin{equation}\label{eq:sec3.1.rem1}
    \begin{aligned}
    \mathrm{d}Z^1_t=\phi_t\mathrm{d}t+\chi_t^{\top}\mathrm{d}W_t,
    \end{aligned}
\end{equation}
for some $Z^1_0\in\RR$. According to \citet{rogers2003duality}, such a structure for the Lagrange multiplier processes is valid too. For the situation above, this specification of $Z^1_t$ does not manage to eliminate the non-concavity of $\eL$'s objective function. In particular, the mere term that changes in the previous derivation(s) is $\phi_t Z^1_t$. Instead of $\phi_t Z^1_t$, under $Z^1_t$ as given in \eqref{eq:sec3.1.rem1}, we would need to work with $\phi_t$ alone. Straightforwardly, this leaves most of the foregoing results unaffected, including the saddle point characterisation of the system of FOC's in \eqref{eq:sec3.1.sys}. We highlight this, because we employ a general semi-martingale process for $Z^1_t$, similar to \eqref{eq:sec3.1.rem1}, in the alternative identification procedure. Second, we observe that the heuristic derivation required $X_T$ to satisfy the following CS condition: $X_T=0$. As a consequence, the identification procedure confirms our assertion in section \ref{sec3}, stating that $X_T=0$ is hidden under the surface of \eqref{eq:sec2.2dynprob}. Third and last, we recall that the integrability conditions imposed upon $Z^0_t$ and $Z^1_t$ must be verified. For the proof of our main result, we stress that this verification step is integrated into the statement of Fenchel Duality. 
\end{remark}

\subsection{Alternative Identification}
In order to tackle the previous issue of non-concavity, we subsequently work with the logarithmic transformation of $h_t$, i.e.~$\log h_t$. The verification scheme of \citet{KleinRoger:2007:Dualityinoptimal} relies on control processes that attain values in conic subsets of $\RR$. This is obviously not the case for $\log h_t$. Therefore, a problem involving $\log h_t$ cannot be subsumed under their method's range of application. Accordingly, our approach is alternative, in the sense that it allows for the inclusion of endogenous constraint processes that achieve values in $\RR$. The corresponding verification scheme is supplied by the notion of Fenchel Duality. Under the specification of $\log h_t$ as a dynamic constraint, the core steps of \citet{KleinRoger:2007:Dualityinoptimal}'s non-adjusted identification procedure remain intact. Next, we consecutively revisit the five corresponding steps, and alter them in conformity with the $\log h_t$-constraint.


In step (i), we introduce $Z^0_t$ from \eqref{eq:sec3.1.zz}, and a new process for $Z^1_t$: 
\begin{equation}\label{eq:sec3.2.znew}
    \begin{aligned}
    \mathrm{d}Z^1_t=-\nu_t\mathrm{d}t+\chi_t^{\top}\mathrm{d}W_t,
    \end{aligned}
\end{equation}
for some $Z^1_0\in\RR$, and two processes $\nu_t$ and $\chi_t$ that are to be determined. Consistent with section \ref{sec4.1}, we assume that $\chi_t$ is regular-enough, such that a stochastic integral involving this term vanishes in expectation. The processes $Z^0_t$ and $Z^1_t$ force the dynamics upon $X_t$ and $\log h_t$, respectively. In contrast to $Z^1_t$ in \eqref{eq:sec3.1.zz}, we now assume that $Z^1_t$ is a infinite variation process that takes on values in the entirety of $\RR$. Note that this is intuitive from a Lagrangian perspective, as $\log h_t$ also achieves values on the entire real line, $\RR$.

In step (ii), we only compute the dynamics of $Z^1_t\log h_t$. As the dynamics of $X_tZ^0_t$ from section \ref{sec4.1} remain unchanged, the identity in \eqref{eq:sec3.1.exp1} still holds. Obviously, the expression in \eqref{eq:sec3.1.exp2} must be altered. In agreement with the procedure outlined in step (ii) of section \ref{sec4.1}, we are able to modify this expression into the following equation:   
\begin{equation}\label{eq:sec3.2.exp2}
    \EE\left[-Z^1_T\log h_T-\int_0^T\left[ Z^1_t\left(\beta\log c_t -\alpha\log h_t\right)-\nu_t \log h_t\right]\mathrm{d}t\right]=0.
\end{equation}

In step (iii), we are then able to assemble the Lagrangian functional corresponding to the problem in \eqref{eq:sec2.2dynprob}. Concretely, we insert the identity in \eqref{eq:sec3.1.exp1} along with the one in \eqref{eq:sec3.2.exp2} into the primal objective function. The Lagrangian, $\eL$, then originates from optimising the resulting expression over all endogenous primal control variables $\left(c_t,h_t,X_t,\pi_t\right)\in\RR^3_+\times\RR^N$. Note that this is conceptually equivalent to the derivation of $\eL$ in \eqref{eq:sec3.1.lagr}. Hence, it can be shown that the Lagrangian functional, $\eL$, reads as follows:
\begin{equation}\label{eq:sec3.2.lagr}
\begin{aligned}
    \eL=\sup_{\left\lbrace {c}_t\right\rbrace,\left\lbrace {h}_t\right\rbrace,\left\lbrace X_t\right\rbrace,\left\lbrace \pi_t\right\rbrace}&\EE\Big[-Z^0_TX_T+Z^1_T\log h_T+\int_0^T\bigg\{U\left(t,\frac{{c}_t}{{h}_t}\right)\\&+\left(r_t+a_t+\pi_t^{\top}\sigma_t\left[\lambda_t+b_t\right]\right)Z^0_tX_t-c_tZ^0_t\\&-\left[ Z^1_t\left(\beta \log c_t -\alpha\log h_t\right)-\nu_t \log h_t\right]\bigg\}\mathrm{d}t\Big]+Z^0_0X_0.
\end{aligned}
\end{equation}

In step (iv), we can use the optimisation problem spelled out by $\eL$ to specify the relevant CS conditions. In view of the fact that the dynamics corresponding to $X_tZ^0_t$ in \eqref{eq:sec3.2.lagr} are identical to those in \eqref{eq:sec3.1.lagr}, we find the same requirements for $X_T$, $a_t$ and $b_t$. That is, since $Z^0_T\geq 0$, $X_tZ^0_t\geq 0$ and $\pi_t\in\RR^N$, we must impose $X_T=0$, $a_t\leq -r_t$ and $b_t=-\lambda_t$ to ensure that $\eL<\infty$. As for the dynamics corresponding to $Z^1_t\log h_t$, we note that $\log h_t$ is stochastic and $\RR$-valued. Hence, to retain finiteness of $\eL$, we have to require that $Z^1_T=0$ holds. We emphasise that such a CS condition is the mere way of bounding the Lagrangian functional, $\eL$, in the presence of $\RR$-valued constraint processes. 

Ultimately, optimising $\eL$'s objective over $\left(c_t,h_t\right)\in\RR_+$ yields the following FOC's:
\begin{equation}\label{eq:sec3.2.sys}
    \begin{aligned}
    &U'\left(t,\frac{{c}_t}{{h}_t}\right)\frac{1}{{h}_t}-Z_t^0-\beta\frac{1}{{c}_t}Z_t^1=0,\\ & U'\left(t,\frac{{c}_t}{{h}_t}\right)\frac{1}{{h}_t}-\nu_t\frac{1}{c_t}-\alpha Z_t^1 \frac{1}{c_t}=0.
    \end{aligned}
\end{equation}
In contrast to the system of equations in \eqref{eq:sec3.1.sys}, the latter set of FOC's allows for an explicit solution in terms of $\left(c_t,h_t\right)$. To be more precise, on account of equivalence arguments, it is clear that $Z^0_t=\big(\nu_t+\left[\alpha-\beta\right]Z^1_t\big)\frac{1}{c_t^{*}}$ should hold. Re-inserting the identity for $c_t^{*}$ into either of the two FOC's implies an explicit solution for $h_t^{*}$. For the purpose of notation-related clarity, we endow these dual-feasibility conditions with an asterisk. The solutions to the system of two equations presented in \eqref{eq:sec3.2.sys} then live by: 
\begin{equation}\label{eq:sec3.2.sols}
    c_t^{*}=\frac{\nu_t+\left[\alpha-\beta\right]Z^1_t}{Z^0_t}\quad\mathrm{and}\quad h_t^{*}=\frac{c_t^{*}}{I\left(t,\nu_t+\alpha Z^1_t\right)}.
\end{equation}
Note that we have $h_t>0$ by construction, cf. \eqref{eq:sec2.2.habsol}. Additionally, because the agent derives utility from a function defined on the domain $\left[0,T\right]\times\RR_+$, we likewise know that $c_t>0$ must hold. Therefore, in order to assure that $c_t^{*},h_t^{*}>0$ hold true, we must impose that $\nu_t>-\alpha Z^1_t$ and $\nu_t>-\left[\alpha-\beta\right]Z^1_t$. These requirements can be combined as follows: $\nu_t>\max\left\lbrace -\alpha Z^1_t,-\left[\alpha-\beta\right]Z^1_t\right\rbrace$. Now, we note that $Z^1_T=0$ must hold, such that $Z^1_t=\EE\big[\int_t^T\nu_s\mathrm{d}s \ \big| \ \eF_t\big]$ can be shown to solve the SDE in \eqref{eq:sec3.2.znew}. For the special $\alpha=\beta$ case, it is easy to demonstrate that the condition imposed on $\nu_t$ reduces to $\nu_t>0$, due to which $Z^1_t>0$ must hold. With regard to the $\chi_t$ process, we observe that its specification is incorporated into $Z^1_t$'s definition, cf. Remark \ref{rem3} for details.

To demonstrate that \eqref{eq:sec3.2.sols} specifies a global maximum of $\eL$, let us introduce :
\begin{equation}\label{eq:sec3.2.function}
    F\left(x,y\right)=  U\left(t,\frac{x}{y}\right)-xz_0- z_1\left(\beta\log x - \alpha\log y\right) +\nu \log y,
\end{equation}
which describes a function $F:\RR^2_+\rightarrow\RR$. Similar to $L$ in \eqref{eq:sec3.1.funcl}, for $F$, we keep the constants $z_0,t\in\RR_+$ and $\nu,z_1\in\RR$ fixed and assume that $\nu>\max\left\lbrace-\alpha z_1,-\left[\alpha-\beta\right]z_1\right\rbrace$ holds. Suppose that $G:\RR^2\rightarrow\RR$ defines a function as follows: $G\left(v,w\right)=F\left(e^v,e^w\right)$ for all $v,w\in\RR$. One is able to show that the determinant of the Hessian of $G$, say $\mathrm{det}\left(H_G\right)$, equates to: $\mathrm{det}\left(H_G\right)=-e^{2v-w}z_0\big(U''\left(t,e^{v-w}\right)e^{v-w}+U'\left(t,e^{v-w}\right)\big)$. As $-x\frac{U''\left(t,x\right)}{U'\left(t,x\right)}>1$ for all $t\in\left[0,T\right]$ and $x\in\RR_+$ by assumption, it is evident that $\mathrm{det}\left(H_G\right)>0$ holds.  Furthermore, the second partial derivative of $G$ in the $v$-direction is characterised by: $\frac{\partial^2 G}{\partial v\partial v}=\big(U''\left(t,e^{v-w}\right)e^{v-w}+U'\left(t,e^{v-w}\right)\big)e^{v-w}-e^vz_0$. Obviously, $\frac{\partial^2 G}{\partial v\partial v}<0$ holds. As $H_G$ is consequently negative definite, the function $G$ is globally concave in $\left(v,w\right)$. We then note that $G$ achieves a stationary point where $v^{*}$ and $w^{*}$ satisfy the following equations:
\begin{equation}\label{eq:sec3.2.gstat}
    e^{v^{*}}=\frac{\nu+\left[\alpha-\beta\right]z_1}{z_0}\quad\mathrm{and}\quad e^{w^{*}}= \frac{e^{v^{*}}}{I\left(t,\nu+\alpha z_1\right)}.
\end{equation}
Since $G$ is globally concave in $\left(v,w\right)$, the stationary point $\left(v^{*},w^{*}\right)$, which solves \eqref{eq:sec3.2.gstat}, specifies a global maximum of $G$. As a result, $G\left(v^{*},w^{*}\right)\geq G\left(v,w\right)$ holds for all $v,w\in\RR$. Therefore, $F\left(e^{v^{*}},e^{w^{*}}\right)\geq F\left(x,y\right)$ is true, for all $x,y>0$. Setting $x^{*}=e^{v^{*}}$ and $y^{*}=e^{w^{*}}$, it is clear that $x^{*}=\frac{\nu+\left[\alpha-\beta\right]z_1}{z_0}$ and $y^{*}=\frac{e^{v^{*}}}{I\left(t,\nu+\alpha z_1\right)}$ globally maximise $F$. Note that $\left(x^{*},y^{*}\right)$ also defines the stationary point of $F$, i.e.~$x^{*}$ and $y^{*}$ jointly solve the following two equations: $U'\big(t,\frac{x^{*}}{y^{*}}\big)\frac{1}{y^{*}}-z_0-\frac{\beta z_1}{x^{*}}=0$ and $-U'\big(t,\frac{x^{*}}{y^{*}}\big)\frac{x^{*}}{y^{*^{2}}}+\frac{\nu+\alpha z_1}{y^{*}}=0$. The FOC's in \eqref{eq:sec3.2.sols} correspond to the state-wise stationary point of $F\left(c_t,h_t\right)$ for $z_0=Z^0_t$, $z_1=Z^1_t$ and $\phi=\phi_t$, at all $t\in\left[0,T\right]$. Hence, the solutions of \eqref{eq:sec3.2.sys}, presented in \eqref{eq:sec3.2.sols}, globally maximise $\eL$'s objective. Correspondingly, the expression for the reduced Lagrangian supplies a well-posed candidate for the dual objective, as it renders upper bound on $\EE\big[\int_0^TU\big(t,\frac{c_t}{h_t}\big)\mathrm{d}t\big]$. 

In step (v), we minimise this reduced Lagrangian over $Z^0_0$ and $\nu_t$:
\begin{equation}\label{eq:sec3.2.dual}
\begin{aligned}
    \inf_{\nu_t\in\Psi,Z^0_0\in\RR_+}&\EE\left[\int_0^T\left\lbrace -V_1\left(t,\nu_t+\alpha\EE\left[\int_t^T\nu_s\mathrm{d}s\cond\eF_t\right]\right)\right.\right.\\&\left.\left.-Z^0_0M_tV_2\left(\frac{\nu_t+\left(\alpha-\beta\right)\EE\left[\int_t^T\nu_s\mathrm{d}s\cond\eF_t\right]}{Z^0_0M_t}\right)\right\rbrace\mathrm{d}t\right]+Z^0_0X_0.
\end{aligned}
\end{equation}
Here, we define $V_1\left(x\right)=-U\left(t,I\left(t,x\right)\right)+x\log I\left(t,x\right)$ and $V_{2}\left(x\right)=x-x\log x$, for all $t\in\left[0,T\right]$ and $x\in\RR_+$. Furthermore, we let $\Psi$ be some convex space, to be determined. As in \citet{KleinRoger:2007:Dualityinoptimal}, we use here that the reduced Lagrangian is decreasing in $Z^0_t$ to set $a_t=-r_t$, rather than $a_t\leq -r_t$. As a consequence, $Z^0_t$ coincides up to its starting value, $Z^0_0\in\RR_+$, with the SPD in \eqref{eq:sec2.1.spd}, i.e.~$Z^0_t=Z^0_0M_t$ holds. Furthermore, we employ that $Z^1_t$ is entirely characterised in terms of $\left\lbrace\nu_s\right\rbrace_{s\in\left[t,T\right]}$, because $Z^1_t=\EE\big[\int_t^T\nu_s\mathrm{d}s \ \big| \ \eF_t\big]$ is true. Hence, as accounted for in the infimum's specification of \eqref{eq:sec3.2.dual}, the mere ``free'' dual variables are $\nu_t$ and $Z^0_0$. According to the slightly adjusted identification scheme of \citet{KleinRoger:2007:Dualityinoptimal}, the problem in \eqref{eq:sec3.2.dual} defines a proper candidate for the dual formulation of \eqref{eq:sec2.2dynprob}. Usage of the logarithmic transformation of $h_t$ as a constraint process, hence, manages (i) to tackle the inability to spell out analytical expression for the relevant FOC's, and (ii) to bypass non-concave specifications of the Lagrangian objective function.  

\begin{remark}\label{rem3}
As in Remark \ref{rem2}, we would like to make three observations. First, suppose that we define a new process, $\psi_t$, based on the dual control in \eqref{eq:sec3.2.dual}, $\nu_t$:
\begin{equation}\label{eq:sec3.2.rem1}
    \begin{aligned}
    \psi_t=\nu_t+\alpha\EE\left[\int_t^T\nu_s\mathrm{d}s\cond\eF_t\right],
    \end{aligned}
\end{equation}
for all $t\in\left[0,T\right]$. Then, we are able to regard the identity in \eqref{eq:sec3.2.rem1} as a Volterra equation for $\nu_t$ with the following solution: $\nu_t=\psi_t-\alpha\EE\big[\int_t^Te^{-\alpha\left(s-t\right)}\psi_s\mathrm{d}s \ \big| \ \eF_t\big]$, for all $t\in\left[0,T\right]$. As a consequence, one is able to re-define the candidate dual formulation in \eqref{eq:sec3.2.dual} in terms of $\psi_t$ alone. In that case, the second argument of $\left(t,x\right)\mapsto V_1\left(t,x\right)$ would read $\psi_t$; the argument of $x\mapsto V_2\left(x\right)$ would read $\frac{1}{Z^0_0M_t}\big(\psi_t-\beta\EE\big[\int_t^Te^{-\alpha\left(s-t\right)}\psi_s\mathrm{d}s\ \big| \ \eF_t\big]\big)$. The candidate formulation that follows from this re-definition is supplied in Theorem \ref{thm4.1} using Fenchel Duality. Note that these different formulations are identical to each other, i.e.~the identification procedure indeed generates a well-posed dual problem. Second, we stress that the $\chi_t$ process in the SDE for $Z^1_t$ in \eqref{eq:sec3.2.znew} does not play a role in the candidate formulation. Nevertheless, one is always able to obtain an explicit specification for this process by means of a simple application of It\^{o}'s Lemma to $Z^1_t$. Third and last, we observe that this alternative identification procedure does not work for the additive habit model. However, the standard variant of \citet{KleinRoger:2007:Dualityinoptimal}'s method, as employed in section \ref{sec4.1}, manages to generate a proper candidate dual for the additive configuration. This can be attributed to the structure of the problem in \eqref{eq:sec3.1.dynprob.additive}, cf.~Remark \ref{rem1}.
\end{remark}

\section{Dual Formulation}\label{sec4}
In this section, we provide the main result of this paper: the dual formulation of the optimal consumption problem in \eqref{eq:sec2.2dynprob}. We divide this section into three parts. First, we present the dual formulation and formalise that it satisfies strong duality, in Theorem \ref{thm4.1}. Observe that this dual formulation coincides with the candidate problem identified in \eqref{eq:sec3.2.dual}. In the same part, we briefly comment on the proof required to obtain the dual formulation. Second, we present the complete mathematical proof of the corresponding strong duality result. Third and last, we discuss certain implications of the aforementioned result concerning the analytical structure of the optimal primal and dual controls. 

\subsection{Main Result: Strong Duality}
Theorem \ref{thm4.1} contains the main result of this paper. Its statement formalises the fact that the optimal (dual) control problem in \eqref{eq:sec4.2.dual} and the optimal consumption problem in \eqref{eq:sec2.2dynprob} are dual to each other, satisfying strong duality. First, we provide the theorem itself. Second, we comment on its corresponding proof. 
\begin{theorem}\label{thm4.1}
Consider the optimal consumption problem in \eqref{eq:sec2.2dynprob} and define the primal objective function: $J\left(X_0,\left\lbrace c_t,\pi_t\right\rbrace\right)=\EE\big[\int_0^TU\big(t,\frac{c_t}{h_t}\big)\mathrm{d}t\big]$. Furthermore, introduce the following two concave conjugates: $V_1\left(t,x\right)=\inf_{z\in\RR}\left\lbrace -U\left(t,e^{-z}\right)-xz\right\rbrace=-U\left(t,I\left(t,x\right)\right)+x\log I\left(t,x\right)$ and $V_2\left(x\right)=\inf_{z\in\RR}\left\lbrace e^z-xz\right\rbrace=x-x\log x$, for all $t\in\left[0,T\right]$ and $x\in\RR_+$. Then, the dual formulation of the optimal consumption problem in \eqref{eq:sec2.2dynprob} is given by:
\begin{equation}\label{eq:sec4.2.dual}
\begin{aligned}
    \inf_{\psi_t\in L^2\left(\Omega\times\left[0,T\right]\right),\eta\in\RR_+}&\EE\left[\int_0^T\left\lbrace -{V}_1\left(t,\psi_t\right) \right.\right. \\&\left.\left.-\eta M_t{V}_2\left(\frac{\psi_t-\beta\EE\left[\int_t^Te^{-\alpha\left(s-t\right)}\psi_s\mathrm{d}s\cond\eF_t\right]}{\eta M_t}\right)\right\rbrace\mathrm{d}t\right]+\eta X_0.
\end{aligned}
\end{equation}
Suppose that $\mathcal{V}\left(X_0,\psi_t,\eta \right)$ represents the dual objective function of \eqref{eq:sec4.2.dual}. Then, the problems in \eqref{eq:sec2.2dynprob} and \eqref{eq:sec4.2.dual} satisfy strong duality, for all $X_0\in\RR_+$: 
\begin{equation}\label{eq:thm4.1.eq1}
    \sup_{\left\lbrace c_t,\pi_t\right\rbrace_{t\in\left[0,T\right]}\in{\mathcal{A}}_{X_0}}J\left(X_0,\left\lbrace c_t,\pi_t\right\rbrace\right)=\inf_{\psi_t\in L^2\left(\Omega\times\left[0,T\right]\right),\eta \in\RR_+}\mathcal{V}\left(X_0,\psi_t,\eta\right).
\end{equation}
\end{theorem}
\begin{proof}
The proof is given in section \ref{appendix:a}. 
\end{proof}

Typically, the Legendre transform alone suffices to establish a strong duality result. However, due to the non-concavity and path-dependency of the objective of \eqref{eq:sec2.2dynprob}, the Legendre transform cannot be used to derive strong duality. We have extensively analysed these duality-linked issues in section \ref{sec3.3}. Therefore, to prove Theorem \ref{thm4.1}, we apply a change of variables and employ Fenchel Duality, cf.~Proposition \ref{fenchelduality}. This form of duality can be regarded as a generalisation of the Legendre result to problems involving path-dependent linear transformations of one of the control variables. On the basis of Fenchel Duality, deriving strong duality for problems \eqref{eq:sec2.2dynprob} and \eqref{eq:sec4.2.dual} is straightforward. First, we re-express the primal problem \eqref{eq:sec2.2dynprob} in terms of its static equivalent and $\log c_t$. Second, we use Fenchel Duality to demonstrate that strong duality holds for the static problem and $\inf_{\psi_t\in L^2\left(\Omega\times\left[0,T\right]\right)}\mathcal{V}\left(X_0,\psi_t,\eta\right)$. Third and last, we resort to a technical argument (Lemma \ref{lemma1}) in order to extend this strong duality result to \eqref{eq:sec2.2dynprob} and \eqref{eq:sec4.2.dual}. To conclude, let us emphasise that the dual formulations in \eqref{eq:sec3.2.dual} and \eqref{eq:sec4.2.dual} are identical. For details on the latter result, we refer to our first observation included in Remark \ref{rem3}.
\begin{remark}\label{rem4}
To recover the no-habit case from Theorem \ref{thm4.1}, it suffices to fix $\alpha=\beta=0$. Setting $\alpha=\beta=0$ in \eqref{eq:sec4.2.dual} provides us with the following dual formulation: 
\begin{equation}\label{eq:sec4.2.dual2}
\begin{aligned}
    \inf_{\psi_t\in L^2\left(\Omega\times\left[0,T\right]\right),\eta\in\RR_+}&\EE\left[\int_0^T\left\lbrace -{V}_1\left(t,\psi_t\right) -\eta M_t{V}_2\left(\frac{\psi_t}{\eta M_t}\right)\right\rbrace\mathrm{d}t\right]+\eta X_0.
\end{aligned}
\end{equation}
Here, the dual forces $\left\lbrace \psi_t\right\rbrace_{t\in\left[0,T\right]}$ to satisfy $\psi_t>0$ for all $t\in\left[0,T\right]$. In line with the exclusion of $h_t$ in the primal, the no-habit dual does not contain the $\EE\big[\int_t^Te^{-\alpha\left(s-t\right)}\psi_s\mathrm{d}s\ \big| \ \eF_t\big]$ term. The dual in \eqref{eq:sec4.2.dual2} differs from the conventional one in e.g.~\citet{cvitanic1992convex}. Note that this is not troublesome, as the dual formulations for convex optimisation problems are not unique, cf.~\citet{rockafellar2015convex}. In fact, after inserting the optimal dual control, say $\psi_t^{\mathrm{opt}}$, into \eqref{eq:sec4.2.dual2}, we find the conventional formulation: $\inf_{\eta\in\RR_+}\EE\big[\int_0^TV\big(t,\eta M_t\big)\mathrm{d}t\big]$, where $V\left(t,x\right)=\sup_{x\in\RR_+}\big\{U(t,z)-xz\big\}$. For the analytical specification of the optimal no-habit dual control, $\psi_t^{\mathrm{opt}}$, we refer the reader to Example \ref{ex1}. The aforementioned difference is attributable to the fact that this dual ensues from an application of Fenchel Duality. Namely, this notion of duality involves two convex conjugates instead of one. Moreover, it requires one to re-express the primal control as follows: $c_t=e^{-\left(-\log c_t\right)}$ for all $t\in\left[0,T\right]$. Due to these two features, the dual accommodates two functions, $V_1$ and $V_2$ that coincide with the concave conjugates of $x\mapsto -U\left(t,e^{-x}\right)$ and $x\mapsto e^x$, respectively.
\end{remark}

\subsection{Proof of Theorem \ref{thm4.1}}\label{appendix:a}
In this section, we provide the proof of Theorem \ref{thm4.1}. To prove that \eqref{eq:sec4.2.dual2} is the dual formulation of \eqref{eq:sec2.2dynprob}, we make use of Fenchel Duality as formalised in Theorem 4.3.3 of the textbook by \citet{borwein2004techniques}. As this theorem lies at the heart of our proof, we provide its statement in the following proposition. 
\begin{proposition}\label{fenchelduality}
Let $f:X\rightarrow\RR\cup\left\lbrace \infty\right\rbrace$ and $g:Y\rightarrow\RR\cup\left\lbrace \infty\right\rbrace$ be two continuous and convex functions. Additionally, introduce the bounded linear map $A:X\rightarrow Y$. Here, $X$ and $Y$ outline two Banach spaces. Then, the Fenchel problems are given by:
\begin{equation}\label{eq:fenchel.eq1}
    \begin{aligned}
    p^{*}&=\inf_{x\in X}\left\lbrace f\left(x\right)+g\left(Ax\right)\right\rbrace\\
    d^{*}&=\sup_{y^{*}\in Y}\left\lbrace -f^{*}\left(A^*y^*\right)-g^*\left(-y^*\right)\right\rbrace,
    \end{aligned}
\end{equation}
and satisfy weak duality, $d^{*}\leq p^{*}$. Here, $f^{*}$ and $g^{*}$ represent the convex conjugates of $f$ and $g$, respectively, i.e.~$f^{*}\left(x\right)=\sup_{z\in X}\left\lbrace \left\langle x,z\right\rangle-f\left(z\right)\right\rbrace$ and $g^{*}\left(y\right)=\sup_{z\in Y}\left\lbrace \left\langle y,z\right\rangle-g\left(z\right)\right\rbrace$, for all $x\in X^{*}$ and $y\in Y^{*}$. Note that $X^{*}$ and $Y^{*}$ are the dual spaces of $X$ and $Y$, respectively. Moreover, $A^{*}$ is the adjoint of $A$. Strong duality, i.e.~$p^{*}=d^{*}$, holds if either of the following conditions is fulfilled:\\
\begin{enumerate}[label=(\roman*)]
    \item $0\in\mathrm{core}\left(\mathrm{dom} \ g- A\ \mathrm{dom} \ f\right)$ and $f$ and $g$ are both lower semi-continuous. Here, $\mathrm{core}$ stands for the algebraic interior, and $\mathrm{dom} \ h$ is given by $\mathrm{dom} \ h=\left\lbrace z\mid h\left(z\right)<\infty\right\rbrace$ for any function $h$;\\
    \item $A\ \mathrm{dom} \ f\cap \mathrm{cont} \  g\neq\phi$, where $\mathrm{cont}$ are the points where the function is continuous.\\ 
\end{enumerate}
Moreover, if $\left|d^{*}\right|<\infty$ holds, then the supremum in \eqref{eq:fenchel.eq1} is attained.
\end{proposition}
\begin{proof}
See page 136 of \citet{borwein2004techniques}.
\end{proof}

We continue by aligning the notation of our primal and dual problems with the notation of Proposition \ref{fenchelduality}. To this end, we start by deriving an alternative representation of \eqref{eq:sec2.2dynprob}. This alternative representation is based on the static formulation of optimal investment-consumption problems, due to  \citet{pliska1986stochastic}, \citet{karatzas1987optimal}, and \citet{CoxHuang:1989:Optimalconsumptionand,cox1991variational}. This static counterpart can be obtained by means of Lagrangian machinery. We provide this alternative formulation in the subsequent lemma. 

\begin{lemma}\label{lemma1}
Define the following function: 
\begin{equation}
\begin{aligned}
    \mathcal{J}\left(X_0,-\log c_t,\eta\right)=\EE\left[\int_0^TU\left(t,e^{\log c_t-\log h_t}\right)\mathrm{d}t\right]-\eta\EE\left[\int_0^Te^{\log c_t}M_t\mathrm{d}t\right]+\eta X_0.
\end{aligned}
\end{equation}
Then, for all $X_0\in\RR_+$, the following optimisation problems are identical:
\begin{equation}\label{eq:lem1.main2}
    \sup_{\left\lbrace c_t,\pi_t\right\rbrace_{t\in\left[0,T\right]}\in{\mathcal{A}}_{X_0}}J\left(X_0,\left\lbrace c_t,\pi_t\right\rbrace\right)=\inf_{\eta\in\RR_+}\sup_{-\log c_t\in L^2\left(\Omega\times\left[0,T\right]\right)}\mathcal{J}\left(X_0,-\log c_t,\eta\right).
\end{equation}
\end{lemma}
\begin{proof}
By arguments similar to those that yield Lemma 2.2 in \citet{CoxHuang:1989:Optimalconsumptionand} and Proposition 7.3 in \citet{cvitanic1992convex}, we know that $\left\lbrace c_t,\pi_t\right\rbrace_{t\in\left[0,T\right]}\in{\mathcal{A}}_{X_0}$ if and only if $\left\lbrace c_t\right\rbrace_{t\in\left[0,T\right]}$ satisfies $\EE\left[\int_0^Tc_tM_t\mathrm{d}t\right]\leq X_0$ and $\log c_t\in L^2\left(\Omega\times\left[0,T\right]\right)$. Therefore, maximisation of $J\left(X_0,\left\lbrace c_t,\pi_t\right\rbrace\right)$ over $\left\lbrace c_t,\pi_t\right\rbrace_{t\in\left[0,T\right]}\in{\mathcal{A}}_{X_0}$ is the same as maximisation of $J\left(X_0,\left\lbrace c_t,\pi_t\right\rbrace\right)$ over all $\log c_t\in L^2\left(\Omega\times\left[0,T\right]\right)$ such that $\EE\left[\int_0^Tc_tM_t\mathrm{d}t\right]\leq X_0$ holds. As a result, we are able to derive the following set of equations:
\begin{equation}\label{eq:lem1.eq1}
    \begin{aligned}
    &\sup_{\left\lbrace c_t,\pi_t\right\rbrace_{t\in\left[0,T\right]}\in{\mathcal{A}}_{X_0}} \ \EE\left[\int_0^TU\left(t,\frac{c_t}{h_t}\right)\mathrm{d}t\right]\\&=\sup_{\log c_t\in L^2\left(\Omega\times\left[0,T\right]\right)\ \mathrm{s.t.} \  \EE\left[\int_0^Tc_tM_t\mathrm{d}t\right]\leq X_0} \ \EE\left[\int_0^TU\left(t,\frac{c_t}{h_t}\right)\mathrm{d}t\right]\\&=\inf_{\eta\in\RR_+}\left(\sup_{-\log c_t\in L^2\left(\Omega\times\left[0,T\right]\right)} \left\lbrace\EE\left[\int_0^TU\left(t,\frac{c_t}{h_t}\right)\mathrm{d}t\right]-\eta \EE\left[\int_0^Tc_tM_t\mathrm{d}t\right]+\eta X_0\right\rbrace\right).
    \end{aligned}
\end{equation}

The last equality is a result of the following ingredients. First, we know that $c_t={X_0\epsilon}\big({\EE\big[\int_0^TM_t\mathrm{d}t\big]}\big)^{-1}$ for $\epsilon\in\left(0,1\right)$ is a strictly feasible solution to the static formulation of the consumption problem. Second, we have that $h_t>0$ and $c_t>0$. Hence, $c_t=e^{\log c_t}$ and $h_t=e^{\log h_t}$. Using this, we derive that $\EE\big[\int_0^TU\big(t,\frac{c_t}{h_t}\big)\mathrm{d}t\big]$ is strictly concave in $-\log c_t\in L^2\left(\Omega\times\left[0,T\right]\right)$. Similarly, we have that $\eta \EE\big[\int_0^Tc_tM_t\mathrm{d}t\big]$ is strictly convex in $-\log c_t\in L^2\left(\Omega\times\left[0,T\right]\right)$. Third, by the assumptions imposed upon $U$, and the fact that $\log c_t,\log h_t\in L^2\left(\Omega\times\left[0,T\right]\right)$, it holds that $\EE\big[\int_0^TU\big(t,\frac{c_t}{h_t}\big)\mathrm{d}t\big]<\infty$. These properties validate the last equality, cf.~Theorem 1 on page 217 of \citet{luenberger1997optimization}. The step from \eqref{eq:lem1.eq1} to \eqref{eq:lem1.main2} is trivial using that $c_t=e^{\log c_t}$ and $h_t=e^{\log h_t}$.
\end{proof}

To align our notation with the one of Proposition \ref{fenchelduality}, we should have:
\begin{equation}
    d^{*}=\sup_{-\log c_t\in L^2\left(\Omega\times\left[0,T\right]\right)}\mathcal{J}\left(X_0,-\log c_t,\eta\right).
\end{equation}
Accordingly, in the nomenclature of the aforementioned proposition, we have that $y^{*}=-\log c_t$ and $Y=L^2\left(\Omega\times\left[0,T\right]\right)$, which is a Banach space. Moreover, in terms of the functions $f^{*}$ and $g^{*},$ and the mapping $A$, we must have the following: 
\begin{equation}\label{eq:f.eq1}
\begin{aligned}
   -f^{*}\left(A^*y^*\right)&=\EE\left[\int_0^TU\left(t,e^{-A^{*}\left(-\log c_t\right)}\right)\mathrm{d}t\right]\\
    -g^*\left(-y^*\right)&=-\eta \EE\left[\int_0^Te^{\log c_t}M_t\mathrm{d}t\right]+\eta X_0,
\end{aligned}
\end{equation}
where the linear map $A^{*}$ is given by:
\begin{equation}
    A^{*}\left(-\log c_t\right)=-\log c_t+\beta \int_0^te^{-\alpha \left(t-s\right)}\log c_s\mathrm{d}s.
\end{equation}
Clearly, $A^{*}:L^2\left(\Omega\times\left[0,T\right]\right)\rightarrow L^2\left(\Omega\times\left[0,T\right]\right)$. Therefore, by adjointess arguments, we must have that $A:L^2\left(\Omega\times\left[0,T\right]\right)\rightarrow L^2\left(\Omega\times\left[0,T\right]\right)$, too. 

According to the equations for $f^{*}$ and $g^{*}$ in \eqref{eq:f.eq1}, we ought to have: 
\begin{equation}
    \begin{aligned}
    f^{*}\left(x\right)&=-\EE\left[\int_0^TU\left(t,e^{-x_t}\right)\mathrm{d}t\right]\\
    g^{*}\left(x\right)&=\eta \EE\left[\int_0^Te^{x_t}M_t\mathrm{d}t\right]-\eta X_0.
    \end{aligned}
\end{equation}
To obtain the definitions of $f$ and $g$ in \eqref{eq:fenchel.eq1}, we recall that: $f^{*}\left(x\right)=\sup_{z\in X}\left\lbrace \left\langle x,z\right\rangle-f\left(z\right)\right\rbrace$ and $g^{*}\left(y\right)=\sup_{z\in Y}\left\lbrace \left\langle y,z\right\rangle-g\left(z\right)\right\rbrace$, for all $x\in X^{*}$ and $y\in Y^{*}$. Correspondingly, in our case: $-f\left(z_1\right)=\inf_{x\in X^{*}}\left\lbrace f^{*}\left(x\right)-\left\langle x,z_1\right\rangle\right\rbrace$ and $-g\left(z_2\right)=\inf_{y\in Y^{*}}\left\lbrace g^{*}\left(y\right)-\left\langle y,z_2\right\rangle\right\rbrace$, for all $z_1\in X$ and $z_2\in Y$. It is then easy to derive the following definitions of $f$ and $g$:
\begin{equation}
\begin{aligned}
    f\left(x\right)&=-\EE\left[\int_0^TV_1\left(t,x_t\right)\mathrm{d}t\right]\\
    g\left(x\right)&=-\EE\left[\int_0^T\eta M_t{V}_2\left(\frac{x_t}{\eta M_t}\right)\mathrm{d}t\right]+\eta X_0.
\end{aligned}
\end{equation}

We observe that $X=Y=L^2\left(\Omega\times\left[0,T\right]\right)$, and that the preceding definitions of $f:X\rightarrow\RR\cup\left\lbrace \infty\right\rbrace$ and $g:Y\rightarrow\RR\cup\left\lbrace \infty\right\rbrace$ constitute two continuous and convex functions. Furthermore, we find that $A$, i.e.~the adjoint of $A^{*}$, is given by the following linear mapping:
\begin{equation}
    A x_t = x_t-\beta \EE\left[\int_t^Te^{-\alpha\left(s-t\right)}x_s\mathrm{d}s\cond\eF_t\right].
\end{equation}
Note here that:
\begin{equation}\label{eq:app.ineqset}
\begin{aligned}
    \left\Vert A x_t\right\Vert_{L^2\left(\Omega\times\left[0,T\right]\right)}&\leq \left\Vert x_t\right\Vert_{L^2\left(\Omega\times\left[0,T\right]\right)}+\beta\left\Vert \EE\left[\int_t^Te^{-\alpha\left(s-t\right)}x_s\mathrm{d}s\cond\eF_t\right]\right\Vert_{L^2\left(\Omega\times\left[0,T\right]\right)}\\&\leq \left\Vert x_t\right\Vert_{L^2\left(\Omega\times\left[0,T\right]\right)}+\beta\EE\left[\int_0^T \EE\left[\int_t^Te^{-2\alpha\left(s-t\right)}x_s^2\mathrm{d}s\cond\eF_t\right]\mathrm{d}t\right]^{\frac{1}{2}}\\&=\left\Vert x_t\right\Vert_{L^2\left(\Omega\times\left[0,T\right]\right)}+\frac{1}{2}\frac{\beta}{\alpha}\left\Vert x_t\left(1-e^{-2\alpha t}\right)^{\frac{1}{2}}\right\Vert_{L^2\left(\Omega\times\left[0,T\right]\right)}\leq \frac{3}{2}\left\Vert x_t\right\Vert_{L^2\left(\Omega\times\left[0,T\right]\right)}.
\end{aligned}
\end{equation}
The first inequality is due to the triangle inequality; the second inequality is a result of H\"{o}lder's inequality; the final inequality is trivial ($1-e^{-2\alpha t}<1$ for all $t\in\left[0,T\right]$). As a consequence of \eqref{eq:app.ineqset}, we know that $A:X\rightarrow Y$ is a bounded linear map. 

Considering Proposition \ref{fenchelduality}, we note that $A\ \mathrm{dom} \ f \ \cap \ \mathrm{cont} \  g=\left(L^2\left(\Omega\times\left[0,T\right]\right)\cap \RR\right)\cap \left(L^2\left(\Omega\times\left[0,T\right]\right)\cap \RR_+\right)\neq\phi$. Hence, by Proposition \ref{fenchelduality}, we have strong duality, which finalises -- via Lemma \ref{lemma1} -- the proof of Theorem \ref{thm4.1}:  
\begin{equation}
    d^{*}=\sup_{-\log c_t\in L^2\left(\Omega\times\left[0,T\right]\right)}\mathcal{J}\left(X_0,-\log c_t,\eta\right)=p^{*}=\inf_{\psi_t\in L^2\left(\Omega\times\left[0,T\right]\right)}\mathcal{V}\left(X_0,\psi_t,\eta\right).
\end{equation}

\subsection{Duality Relations}\label{sec4.3}
For convex optimisation problems, duality theory can be employed to disclose the relationship between the primal and dual controls, i.e.~the duality relation. This duality relation infers how the primal controls analytically depend on the dual controls, and vice versa. The key characteristic of this relation is that it yields the optimal primal (dual) controls after insertion of the optimal dual (primal) controls (respectively). Therefore, the duality relation contains important information about the analytical structure of the optimal primal and dual variables. In addition to this, it provides an alternative view on the mechanisms that are involved with optimising the primal and dual problems. As the dual in \eqref{eq:sec4.2.dual} follows from Fenchel Duality rather than from the Legendre transform, its implied duality relations differ from the conventional ones. In fact, the duality relations\footnote{These duality relations follow from the fact that the primal and dual objectives, in \eqref{eq:sec2.2dynprob} and \eqref{eq:sec4.2.dual}, are conjugate to each other. This concretely means that these expressions bind in the ``point'' outlined by \eqref{eq:dualrelation}, conditional on $\eta X_0=\eta \EE\big[\int_0^Tc_tM_t\mathrm{d}t\big]$ being true. Note that the duality relation in \eqref{eq:dualrelation} corresponds to the dual in \eqref{eq:sec4.2.dual}. For the alternative, howbeit identical, representation in \eqref{eq:sec3.2.dual}, the duality relations read: $c_t^{*}=\frac{\nu_t+\left(\alpha-\beta\right)\EE\big[\int_t^T\nu_s\mathrm{d}s \ \big| \ \eF_t \big]}{\eta M_t}$ and $\frac{c_t^{*}}{h_t^{*}}=I\big(t, \nu_t+\alpha\EE\big[\int_t^T\nu_s\mathrm{d}s\ \big| \ \eF_t\big]\big)$ for all $t\in\left[0,T\right]$.} for the problems in \eqref{eq:sec2.2dynprob} and \eqref{eq:sec4.2.dual} are for all $t\in\left[0,T\right]$ given by: 
\begin{equation}\label{eq:dualrelation}
    c_t^{*}=\frac{\psi_t-\beta\EE\left[\int_t^Te^{-\alpha\left(s-t\right)}\psi_s\mathrm{d}s\cond\eF_t\right]}{\eta M_t}\quad\mathrm{and}\quad \frac{c_t^{*}}{h_t^{*}}=I\left(t,\psi_t\right)
\end{equation}

In a technical sense, the duality relation for consumption in \eqref{eq:dualrelation}, $c_t^{*}$, can be regarded as a specification of optimal consumption in some auxiliary (artificial) market. To ensure that consumption defined by this relation is admissible and optimal in the true market, the dual problem in \eqref{eq:sec4.2.dual} aims to characterise this identity for $c_t^{*}$ in such a manner that it generates the habit level in \eqref{eq:dualrelation}, $h_t^{*}$. In an economic sense, we note that dual-implied consumption $c_t^{*}$ is endowed with a ``penalty term''. Concretely, selecting high values for $\psi_t$ at future dates, requires one to increase $\psi_t$ today so as to arrive at similar utility levels. This mechanism inversely reflects the agent's viewpoint in the primal problem. Namely, if this agent selects high values for $c_t$ today, via $h_t$, he/she is required to increase $c_t$ even further to maintain similar utility levels. To obtain some insights into the role that $\psi_t$ plays in minimising $\mathcal{V}$, we now conclude with Example \ref{ex1}.
\begin{example}\label{ex1}
Suppose that $\alpha=\beta=0$. Then, minimisation of the dual results in: 
\begin{equation}\label{eq:ex1.eq1}
    \psi_t^{\mathrm{opt}}=\widehat{I}\left(t,\eta^{\mathrm{opt}} M_t\right)\eta^{\mathrm{opt}} M_t,
\end{equation}
for all $t\in\left[0,T\right]$, where $\widehat{I}:\left[0,T\right]\times\RR_+\rightarrow\RR_+$ outlines the inverse of marginal utility: $U'\big(t,\widehat{I}\left(t,x\right)\big)=x$, for all $t\in\left[0,T\right]$ and $x\in\RR_+$. We assume that $M_t$ is such that $\widehat{I}\left(t,\eta^{\mathrm{opt}} M_t\right)\eta^{\mathrm{opt}} M_t\in L^2\left(\Omega\times\left[0,T\right]\right)$. Here, $\psi_t^{\mathrm{opt}}$ denotes the optimal dual process, and $\eta^{\mathrm{opt}}$ represents the corresponding dual-optimal constant. In particular, $\eta^{\mathrm{opt}}$ can be obtained from solving $\EE\big[\int_0^T\frac{\psi_t^{\mathrm{opt}}}{\eta^{\mathrm{opt}} M_t}\mathrm{d}t\big]=X_0$ for $\eta^{\mathrm{opt}}$. From the duality relations provided in \eqref{eq:dualrelation}, we know that $\psi_t^{\mathrm{opt}}$ should generate $c_t^{\mathrm{opt}}$ via $c_t^{*}=\frac{\psi_t-\beta\EE\left[\int_t^Te^{-\alpha\left(s-t\right)}\psi_s\mathrm{d}s\cond\eF_t\right]}{\eta M_t}$. Using that $c_t^{*}=\frac{\psi_t}{\eta M_t}$ for $\alpha=\beta=0$, we therefore find that optimal consumption is given by: 
\begin{equation}\label{eq:ex1.eq2}
    c_t^{\mathrm{opt}}=\frac{\psi_t^{\mathrm{opt}}}{\eta^{\mathrm{opt}} M_t}=\widehat{I}\left(t,\eta^{\mathrm{opt}} M_t\right),
\end{equation}
for all $t\in\left[0,T\right]$. Moreover, in the optimum characterised by $c_t^{\mathrm{opt}}$ and $\psi_t^{\mathrm{opt}}$, the value for $\eta^{\mathrm{opt}}\in\RR_+$ is determined such that $\EE\big[\int_0^Tc_t^{\mathrm{opt}}M_t\mathrm{d}t\big]=X_0$ holds. Hence, it is clear that $c_t^{\mathrm{opt}}$ in \eqref{eq:ex1.eq2} coincides with optimal consumption in the no-habit case ($\alpha=\beta=0$). Consider the well-known power utility function, $U\left(t,x\right)=\frac{x^{1-\gamma}}{1-\gamma}$, for $\gamma>1$ and all $x\in\RR_+$. Then, following \eqref{eq:ex1.eq2}, optimal consumption lives by its conventional definition: $c_t^{\mathrm{opt}}=\left(\eta^{\mathrm{opt}} M_t\right)^{-\frac{1}{\gamma}}$. See for instance \citet{merton1971optimum} for a similar representation of $c_t^{\mathrm{opt}}$.
\end{example}

\section{{Financial Application}}\label{sec6}
{We conclude this paper with a relevant financial application of the duality result in Theorem \ref{thm4.1}. This application concerns the design of a dual-control method for evaluating the accuracy of approximations to the optimal solutions of \eqref{eq:sec2.2dynprob}. As far as we know, such a method has not been developed for frameworks involving multiplicative habit formation. We break this section down into three parts. First, we provide the general evaluation mechanism and comment on related technicalities. Second, apart from the one proposed by \citet{van2020consumptiona}, we present a new approximation. The latter approximation is predicated on the duality relations inherent in our strong duality result. Third, we make use of the evaluation mechanism to analyse the approximations' accuracy. We show that our approximation can outperform \citet{van2020consumptiona}'s.}

\subsection{Dual-Control Method}\label{sec6.1}
{To quantify the accuracy of numerical or closed-form approximations to the optimal solutions of \eqref{eq:sec2.2dynprob}, we develop a novel evaluation mechanism or dual-control method. This mechanism is predicated on the approximating techniques proposed by \citet{bick2013solving} and \citet{kamma2021near}. These techniques are developed for investment-consumption problems formulated in financial markets with trading constraints. Even though our setup does not involve such constraints, the core principle of these methods can be used for the design of our evaluating routine. This principle concretely employs the notion of strong duality to observe that any departure from the optimal primal and/or dual controls results in a duality gap. The magnitude of the duality gap can in turn be interpreted as a reliable indicator for the accuracy of specific primal and/or dual approximations. Adapted to our situation and Theorem \ref{thm4.1}, this means the following. The primal value function, $J\left(X_0,\left\lbrace c_t,\pi_t\right\rbrace\right)$, delivers a \textit{lower bound} on the optimal dual value function for each admissible trading-consumption pair, $\left\lbrace c_t,\pi_t\right\rbrace_{t\in\left[0,T\right]}\in\mathcal{A}_{X_0}$. Likewise, the dual value function, $\mathcal{V}\left(X_0,\psi_t,\eta\right)$, spawns an \textit{upper bound} on the optimal primal value function for each feasible pair, $\left(\eta,\psi_t\right)\in\RR_+\times L^2\left(\Omega\times\left[0,T\right]\right)$. To be more precise, for all $X_0\in\RR_+$, $\left\lbrace c_t,\pi_t\right\rbrace_{t\in\left[0,T\right]}\in\mathcal{A}_{X_0}$ and $\left(\eta,\psi_t\right)\in\RR_+\times L^2\left(\Omega\times\left[0,T\right]\right)$, we have:}
{\begin{equation}\label{eq:sec4.2.ineq}
\begin{aligned}
    J\left(X_0,\left\lbrace c_t,\pi_t\right\rbrace\right)&\leq J\left(X_0,\left\lbrace c_t^{\mathrm{opt}},\pi_t^{\mathrm{opt}}\right\rbrace\right) \\&=\mathcal{V}\left(X_0,\psi_t^{\mathrm{opt}},\eta^{\mathrm{opt}}\right)\leq \mathcal{V}\left(X_0,\psi_t,\eta\right).
\end{aligned}
\end{equation}}
{\indent Theorem \ref{thm4.1} infers that the inequality in \eqref{eq:sec4.2.ineq} binds if both $\left( c_t,\pi_t\right)=\big( c_t^{\mathrm{opt}},\pi_t^{\mathrm{opt}}\big)$ and $\left(\eta,\psi_t\right)=\big(\eta^{\mathrm{opt}},\psi_t^{\mathrm{opt}}
\big)$ hold true. Here, $\big( c_t^{\mathrm{opt}},\pi_t^{\mathrm{opt}}\big)$ and $\big(\eta^{\mathrm{opt}},\psi_t^{\mathrm{opt}}
\big)$ represent the optimal primal and dual control variables, respectively.  As a result of strong duality, the difference between $J$ and $\mathcal{V}$ grows, the farther $\left\lbrace c_t,\pi_t\right\rbrace_{t\in\left[0,T\right]}$ and/or $\left(\eta,\psi_t\right)$ are situated from the optima. We can employ this observation to gauge the accuracy of particular approximations as follows. Suppose that $\left\lbrace c_t',\pi_t'\right\rbrace_{t\in\left[0,T\right]}\in\mathcal{A}_{X_0}$ represents an arbitrary admissible trading-consumption pair. Similarly, assume that $\left(\eta',\psi_t'\right)\in\RR_+\times L^2\left(\Omega\times\left[0,T\right]\right)$ specifies a feasible pair of dual controls. Then, $D\left(X_0\right)=\mathcal{V}\left(X_0,\psi_t',\eta'\right)-J\left(X_0,\left\lbrace c_t',\pi_t'\right\rbrace\right)$ characterises for all $X_0\in\RR_+$ the corresponding non-negative duality gap. The economic interpretation of the magnitude of $D\left(X_0\right)\in\RR_+$ alone is challenging, given its representation as a utilitarian quantity. We therefore compute its monetary equivalent in the form of a welfare loss. This monetary expression of the duality gap, $\mathcal{C}\in\RR_+$, is calculated in the following way:}
{\begin{equation}\label{eq:sec4.2.eq2}
    J\left(X_0,\left\lbrace c_t',\pi_t'\right\rbrace\right)=\mathcal{V}\left(X_0\left[1-\mathcal{C}\right],\psi_t',\eta'\right).
\end{equation}}
{In literal terms, $\mathcal{C}\in\RR_+$ represents the proportion of $X_0$ that one needs to deduct from the initial endowment in $\mathcal{V}$ to ``close'' the duality gap. From a strictly financial perspective, $\mathcal{C}$ can be interpreted as a monetary fee that one pays to gain access to the unknown optimal trading-consumption pair. In \citet{bick2013solving} and \citet{kamma2021near}, a comparable interpretation is adopted. Similarly, the quantity can be considered a measure of welfare loss arising from the implementation of potentially sub-optimal controls. We emphasise that $\mathcal{C}$ grows with the magnitude of $D\left(X_0\right)$, and thus with the difference(s) between $\left\lbrace c_t',\pi_t'\right\rbrace_{t\in\left[0,T\right]}$ and $\left\lbrace c_t^{\mathrm{opt}},\pi_t^{\mathrm{opt}}\right\rbrace_{t\in\left[0,T\right]}$, as well as between $\left(\eta',\psi_t'\right)$ and $\left(\eta^{\mathrm{opt}},\psi_t^{\mathrm{opt}}
\right)$. The magnitude of $\mathcal{C}$ consequently outlines an apt metric for assessing the accuracy of primal and/or dual approximations. It is noteworthy that this evaluation approach is entirely analytical and in principle free from any numerical or computational burden.} 

\subsection{Analytical Approximations}\label{sec6.2}
{To illustrate the effectiveness of our dual-control method, we present two analytical approximations to optimal (ratio) consumption, $\widehat{c}_t^{\mathrm{opt}}$. The first approximation is the one provided in Theorem 3.1 of \citet{van2020consumptiona}. Their approximation is based on a first-order Taylor expansion around the ``point'' $\left\lbrace\widehat{c}_t\right\rbrace_{t\in\left[0,T\right]}=1$ of the budget constraint in the static representation of \eqref{eq:sec2.2dynprob}. Note that this budget constraint assumes the following form: $\EE\big[\int_0^TM_t\widehat{c}_th_t\mathrm{d}t\big]\leq X_0$. The primary motivation for such an expansion is that the habit level closely tracks optimal consumption.  Inspired by this approach, we propose an alternative approximation. This approximation is based on a first-order parameter expansion of $\psi_t^{\mathrm{opt}}$ around $\alpha=\beta=0$. Setting $\alpha=\beta=0$ recovers the no-habit solutions, cf.~Example \ref{ex1}. The underlying rationale is that the optimal controls are roughly equal to the no-habit solutions. As the approximation by \citet{van2020consumptiona} relies on power utility, we henceforth assume that $U\left(t,x\right)=e^{-\delta t}\frac{x^{1-\gamma}}{1-\gamma}$ for all $x\in\RR_+$ and $t\in\left[0,T\right]$. Here, $\delta\in\RR_+$ represents the agent's time-preference parameter and $\gamma\in\RR_+$ stands for the coefficient of relative risk aversion. Observe that $U\left(t,x\right)=e^{-\delta t}\log x$ for all $x\in\RR_+$ and $t\in\left[0,T\right]$ when $\gamma=1$. We present the approximations in the following overview:}\\
\begin{itemize}
    \item {\textit{(Approximation by \citet{van2020consumptiona})}. This approximation, $\left\lbrace\widehat{c}_{BBL,s}\right\rbrace_{t\in\left[0,t\right]}$, is derived from a first-order Taylor expansion of the static  budget constraint around the ``point'' $\left\lbrace\widehat{c}_t\right\rbrace_{t\in\left[0,T\right]}=1$. Let $\eta^{*}_{BBL}\in\RR_+$ be a parameter that ensures $\EE\big[\int_0^TM_t\widehat{c}_{BBL,t}h_{BBL,t}\mathrm{d}t\big]=X_0$. Here, $h_{BBL,t}$ spells out the habit level generated by the approximation $\left\lbrace\widehat{c}_{BBL,s}\right\rbrace_{t\in\left[0,t\right]}$. The approximation then reads:}
    \begin{equation}\label{eq:bblnew}
    \widehat{c}_{BBL,t}=\left(\eta^{*}_{BBL}e^{\delta t}M_t\left\lbrace 1+\beta\EE\left[\int_t^Te^{-\left[\alpha-\beta\right]\left(s-t\right)}\frac{M_s}{M_t}\mathrm{d}s\cond\eF_t\right]\right\rbrace\right)^{-\frac{1}{\gamma}}.
\end{equation}

\item {\textit{(Dual-based approximation).} Our approximation, $\left\lbrace\widehat{c}_{dual,t}\right\rbrace_{t\in\left[0,T\right]}$, is based on a first-order parameter expansion of ${\psi}_t^{\mathrm{opt}}$ around the point $\alpha=\beta=0$. Suppose that $\widehat{P}_t=\beta\int_0^t\big[\log\big\{\left(\eta^{*}_{dual}M_s\right)\left[e^{\delta s}\eta^{*}_{dual}M_s\right]^{-\frac{1}{\gamma}}\big\}+\delta s\big]\mathrm{d}s+\left(\gamma\alpha+\left[{1-\gamma}\right]\beta\right)\big({M_t\left[e^{\delta t}M_t\right]^{-\frac{1}{\gamma}}}\big)^{-1}\times\EE\big[\int_t^T{M_s\left[e^{\delta s}M_s\right]^{-\frac{1}{\gamma}}}\mathrm{d}s\ \big| \ \eF_t\big]$ holds. Here, $\eta^{*}_{dual}\in\RR_+$ will be defined subsequently. In addition to this, set $\theta_t^{*}=\left(\eta^{*}_{dual}M_t\right)\big[\eta^{*}_{dual}e^{\delta t+\widehat{P}_t}M_t\big]^{-\frac{1}{\gamma}}$. The parameter expansion of $\psi_t^{\mathrm{opt}}$ around $\alpha=\beta=0$ then results in the following dual approximation: } 
\begin{equation}
    \psi_t^{*}=\left(\eta^{*}_{dual}M_t\right)\left[\eta^{*}_{dual}e^{\delta t+\widehat{P}_t}M_t\right]^{-\frac{1}{\gamma}}e^{\alpha \EE\left[\int_t^T\frac{\theta_s^{*}}{\theta_t^{*}}\mathrm{d}s\cond\eF_t\right]}
\end{equation}
{In order to retrieve a corresponding approximation to optimal (ratio) consumption, we make use of the second duality relation in \eqref{eq:dualrelation}. By inserting the latter expression for $\psi_t^{*}$ into the right-hand side of this relation, we obtain an analytical expression for the primal equivalent of this dual approximation, i.e.~$\left\lbrace\widehat{c}_{dual,t}\right\rbrace_{t\in\left[0,T\right]}$. To this end, suppose that $Q_{t}=\big({M_t\left[e^{\delta t}M_t\right]^{-\frac{1}{\gamma}}}\big)^{-1}\EE\big[\int_t^T{M_s\left[e^{\delta s}M_s\right]^{-\frac{1}{\gamma}}}\mathrm{d}s\ \big| \ \eF_t\big]-\EE\left[\int_t^T\frac{\theta_s^{*}}{\theta_t^{*}}\mathrm{d}s\cond\eF_t\right]$ holds. Additionally, assume that $P_t=\int_0^t\log\left(\eta^{*}_{dual}e^{\delta s}M_s\right)^{-\frac{1}{\gamma}}\mathrm{d}s+\big({M_t\left(e^{\delta t}M_t\right)^{-\frac{1}{\gamma}}}\big)^{-1}\EE\big[\int_t^T{M_s\left(e^{\delta s}M_s\right)^{-\frac{1}{\gamma}}}\mathrm{d}s \ \big| \  \eF_t\big]$. Furthermore, let $\eta^{*}_{dual}$ be a parameter that ensures $\EE\big[\int_0^TM_t\widehat{c}_{dual,t}h_{dual,t}\mathrm{d}t\big]=X_0$. Here, $h_{dual,t}$ spells out the habit level generated by $\left\lbrace\widehat{c}_{dual,s}\right\rbrace_{t\in\left[0,t\right]}$. The approximation then reads:} 
\begin{equation}\label{eq:prop3.eq1}
    \widehat{c}_{dual,t}=\left(\eta^{*}_{dual}e^{\delta t+\beta P_{t}+\frac{\alpha\gamma}{1-\gamma} Q_{t}}M_t\right)^{-\frac{1}{\gamma}}.
\end{equation}
\end{itemize}

{We conclude this section by making three remarks. First, we note that the approximations are presented in terms of \textit{ratio} consumption. In order to retrieve the corresponding expressions for consumption itself, we can make use of the habit component. More specifically, $h_{i,t}=e^{\beta\int_0^te^{-\left[\alpha-\beta\right]\left(t-s\right)}\log \widehat{c}_{i,s}\mathrm{d}s}$ holds for $i\in\left\lbrace BBL, dual\right\rbrace$. Approximate consumption can then be obtained from the following relation: $c_{i,t}'=\widehat{c}_{i,t}h_{i,t}$ for $i\in\left\lbrace BBL, dual\right\rbrace$. Second, as the approximations satisfy the static budget constraint by construction, we know that there exist matching trading strategies that are admissible. However, for a given admissible investment-consumption pair, the trading strategy does not affect the magnitude of the value function. We therefore omit them from this section. Third and last, to obtain an upper bound for the use of our dual-control method, we are in need of (approximate) dual controls. For this purpose, we can resort to the duality relations in \eqref{eq:dualrelation}. On the grounds of the first relation, we can derive the following for our approximations: $\psi_{i,t}'=\eta^{*}_i M_tc_{i,t}'+\beta\EE\big[\int_t^Te^{-\left[\alpha-\beta\right]\left(s-t\right)}\eta^{*}_i M_sc_{i,s}'\mathrm{d}s\big]$ for $i\in\left\lbrace BBL, dual\right\rbrace$. Here, $\eta^{*}_i$ is for $i\in\left\lbrace BBL, dual\right\rbrace$ determined such that $\EE\big[\int_0^TI\left(t,\psi_{i,s}'\right)M_se^{\beta\int_0^se^{-\left[\alpha-\beta\right]\left(s-u\right)}\log I\left(u,\psi_{i,u}'\right)\mathrm{d}u}\mathrm{d}s \ \big| \ \eF_t\big]\big)\mathrm{d}t\big]=X_0$ holds true. The approximate dual pair $\left(\psi_{i,t}',\eta^{*}_i\right)$ can then be inserted into $\mathcal{V}$, which will render an upper bound on the approximate primal value function. As we will accordingly obtain two upper bounds, in our numerical example we select the smallest one.}  

\subsection{Numerical Results}
{We evaluate the accuracy of the two approximations shown in section \ref{sec6.2}, using the evaluation mechanism of section \ref{sec6.1}. To this end, we set $N=1$ in the market model, $\mathcal{M}$, and fix $r_t=r$, $\sigma_t=\sigma$, $\mu_t=\mu$, where $r,\sigma$ and $\mu$ are constants. Based upon the parameter initialisation in \citet{van2020consumptiona}, we define: $X_0=20$ $T=10$, $\gamma=10$, $\delta=0.03$, $\alpha=\beta=0.1$, $\mu=0.05$, $r=0.01$ and $\sigma=0.2$. In Table \ref{tab1}, we present the upper bounds on the welfare losses ($\mathcal{C}$) associated with the approximations, for different values of $\gamma$, $X_0$, $\alpha=\beta$ and $T$. We compute the welfare losses on the basis of the equality displayed in \eqref{eq:sec4.2.eq2}. The primal value function, $J$, in this identity follows directly from the approximations in section \ref{sec6.2}. As addressed at the end of the latter section, we make use of the first duality relation in \eqref{eq:dualrelation} to compute $\mathcal{V}$. Namely, the primal approximations generate via this relation specific dual approximations. These dual approximations in turn result in different values for the dual value function, $\mathcal{V}$. We choose the smallest $\mathcal{V}$ amongst the ensuing two upper bounds in order to make the duality gap as tight as possible. As a consequence, per set of parameter values, we rely on the same upper bound to compute the welfare losses.} 

\begin{table}[!t]\centering
\begin{adjustbox}{max width=\textwidth}
\begin{threeparttable}
\def\sym#1{\ifmmode^{#1}\else\(^{#1}\)\fi}
\sisetup{table-space-text-post = \sym{***}}
{\small
\begin{tabular}{@{\extracolsep{5pt}}l*{11}{S[table-align-text-post=false]}}
\toprule[2.5pt]
& \multicolumn{3}{c}{Coefficient of risk-aversion ($\gamma$)}  & \multicolumn{3}{c}{Initial endowment ($X_0$)}   \\
\cline{2-4}\cline{5-7}
\addlinespace               
\multicolumn{1}{c}{} &\multicolumn{1}{c}{$6$}&\multicolumn{1}{c}{$10$}&\multicolumn{1}{c}{$14$}&\multicolumn{1}{c}{$10$}&\multicolumn{1}{c}{$20$}&\multicolumn{1}{c}{$30$}  \\
\midrule
\addlinespace
 \multicolumn{1}{c}{$\widehat{c}_{BBL,t}$}& 0.357\%  & 0.204\%   & 0.149\% & 0.094\%    & 0.204\%   & 0.354\%  \\
  \addlinespace
 \multicolumn{1}{c}{$\widehat{c}_{dual,t}$}& 0.225\%   & 0.124\% & 0.091\% & 0.129\%    & 0.124\%  & 0.207\% \\
\addlinespace
\addlinespace
  \toprule
& \multicolumn{3}{c}{Speed of mean-reversion ($\alpha=\beta$)}  & \multicolumn{3}{c}{Trading horizon ($T$)}   \\
\cline{2-4}\cline{5-7}
\addlinespace               
\multicolumn{1}{c}{} &\multicolumn{1}{c}{$0.01$}&\multicolumn{1}{c}{$0.1$}&\multicolumn{1}{c}{$0.2$}&\multicolumn{1}{c}{$1$}&\multicolumn{1}{c}{$10$}&\multicolumn{1}{c}{$20$}  \\
\midrule
\addlinespace
 \multicolumn{1}{c}{$\widehat{c}_{BBL,t}$}& 0.003\%  & 0.204\%  & 0.782\% & 0.039\%    & 0.204\%    & 0.523\% \\
  \addlinespace
 \multicolumn{1}{c}{$\widehat{c}_{dual,t}$}& 0.001\% &   0.123\%   & 0.619\% & 0.005\%    & 0.124\%  & 0.636\% \\
\addlinespace
\addlinespace
\midrule  

\end{tabular}}
\end{threeparttable}
\end{adjustbox}
\caption{\label{tab1}{\textbf{Upper bounds on welfare losses ($\mathcal{C}$).} This table reports the upper bounds on the welfare losses corresponding to the approximate solutions provided in section \ref{sec6.2}, i.e. $\widehat{c}_{BBL,t}$ in \eqref{eq:bblnew} and $\widehat{c}_{dual,t}$ in \eqref{eq:prop3.eq1}. These welfare losses are calculated by solving \eqref{eq:sec4.2.eq2} for $\mathcal{C}$. The lower bounds ($J$) in \eqref{eq:sec4.2.eq2} are computed by inserting the approximate primal solutions into the primal value function. The upper bound ($\mathcal{V}$) in \eqref{eq:sec4.2.eq2} is set equal to the smallest dual value function, implied by the approximations via the first relation in \eqref{eq:dualrelation}. The table expresses $\mathcal{C}$ in terms of percentages (\%), for different values of the four displayed parameters ($\gamma$, $X_0$, $\alpha=\beta$ and $T$), under a baseline initialisation of the parameters. This baseline set is fixed as follows: $X_0=20$, $T=10$, $\gamma=10$, $\delta=0.03$, $\alpha=\beta=0.1$, $\mu=0.05$, $r=0.01$ and $\sigma=0.2$. The results are based on $10,000$ simulations and an Euler scheme with $40$ equidistant time-steps.}}
\end{table}

{Table \ref{tab1} shows that the maximal welfare losses generated by the three approximations vary between 0.001\% and 0.782\%. Bearing in mind that these numbers constitute upper bounds on the true errors, we can conclude that the approximations are near-optimal. With regard to the performance of $\widehat{c}_{BBL,t}$, this finding coincides with the numerical evidence reported in \citet{van2020consumptiona}. For their approximation, under a slightly different initialisation of parameters, these authors namely demonstrate that the corresponding welfare losses take on values between 0.08\% and 1.27\%. Concerning our dual-based approximation, we observe that $\widehat{c}_{dual,t}$ manages to outperform $\widehat{c}_{BBL,t}$ in 7 out of the 9 presented cases.  The two situations in which $\widehat{c}_{BBL,t}$ outperforms $\widehat{c}_{BBL,t}$ are when $T=20$ and $X_0=10$.  To explain this phenomenon, we note that $X_0=\frac{1}{r}\left(1-e^{-rT}\right)\approx T$ roughly implies that $\widehat{c}_t\approx 1$. The latter is a consequence of the static budget constraint, $\EE\big[\int_0^TM_t\widehat{c}_th_t\mathrm{d}t\big]=X_0$. As the approximation by \citet{van2020consumptiona} is predicated on a Taylor expansion around $\left\lbrace\widehat{c}_t\right\rbrace_{t\in\left[0,T\right]}=1$, it is evident that $\widehat{c}_{BBL,t}$ achieves better results in these two situations. Furthermore, since $\widehat{c}_{BBL,t}$ is strongly geared towards setups in which $X_0\approx T$ holds, it is also clear why our approximation is more accurate in the remaining 7 cases, where $X_0\neq T$. The expansions around $\alpha=\beta=0$ are namely oriented towards the no-habit solution, which tends to adapt itself to the ratio $\frac{X_0}{T}$. In summary, these illustrations of our dual-control method not only emphasise the abundant opportunities for refinement in analytical approximations but also highlight the additional insights attainable through the dual formulation of the habit formation problem.} 

\bibliographystyle{apalike}
\bibliography{APbib}
\end{document}